%% file: paper.tex
\documentclass[a4paper,UKenglish,cleveref,autoref]{lipics-v2021}

\input{common-includes}

\title{Box Covers and Domain Orderings for Beyond Worst-Case Join Processing} 

\author{Kaleb Alway}{University of Waterloo, Canada}{kpalway@uwaterloo.ca}{}{}

\author{Eric Blais}{University of Waterloo, Canada}{eric.blais@uwaterloo.ca}{}{}

\author{Semih Salihoglu}{University of Waterloo, Canada}{semih.salihoglu@uwaterloo.ca}{}{}

\authorrunning{K.~Alway, E.~Blais, and S.~Salihoglu} 

\Copyright{Kaleb Alway, Eric Blais, and Semih Salihoglu}



\category{}

\iftoggle{appendix-online}{}{
\relatedversion{}
\relatedversiondetails[cite={this-arxiv}]{Full Version}{https://arxiv.org/abs/1909.12102}
}

\supplement{}

\iftoggle{appendix-online}{\nolinenumbers}{} 

\iftoggle{appendix-online}{\hideLIPIcs}{}  

\EventEditors{Ke Yi and Zhewei Wei}
\EventNoEds{2}
\EventLongTitle{24th International Conference on Database Theory (ICDT 2021)}
\EventShortTitle{ICDT 2021}
\EventAcronym{ICDT}
\EventYear{2021}
\EventDate{March 23--26, 2021}
\EventLocation{Nicosia, Cyprus}
\EventLogo{}
\SeriesVolume{186}
\ArticleNo{1}

\begin{document}

\iftoggle{appendix-pgs}{\maketitle}{}

\input{abstract}

\input{introduction}
\input{preliminaries}

\input{related}

\input{generating-boxes}

\input{domain-ordering}

\input{future-work}

\input{conclusions}
\iftoggle{appendix-pgs}{\bibliography{paper}}{}


\appendix

\input{appendix}

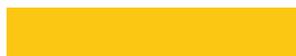
\end{document}

%% file: common-includes.tex
\usepackage[table]{xcolor}
\usepackage{enumerate}
\usepackage{amsfonts}
\usepackage{amsthm}
\usepackage{balance}
\usepackage{bbm}
\usepackage{algpseudocode}
\usepackage{algorithm}
\usepackage{mathtools}
\usepackage{tikz}
\usepackage{pgfplots}
\usepackage{array}
\usepackage{subcaption}
\usepackage{etoolbox}

\newcommand{\AGM}{\text{AGM}}

\newcommand{\attr}{\text{attr}}
\newcommand{\dom}{\text{dom}}
\newcommand{\cb}{\text{cb}}

\newcommand{\Eq}{\text{Eq}}
\newcommand{\DomOrBox}{\text{\texttt{DomOr}}_{\text{\texttt{BoxMinB}}}}
\newcommand{\DomOrCert}{\text{\texttt{DomOr}}_{\text{\texttt{BoxMinC}}}}

\newcommand{\tuplecolour}{blue!50!white}
\newcommand{\tupleopacity}{0.3}
\newcommand{\certcolour}{red}
\newcommand{\boxcolour}{black}
\newcommand{\tc}{\cellcolor{\tuplecolour}}

\renewcommand{\bar}{\overline}
\renewcommand{\tilde}{\widetilde}

\newcommand{\kaleb}[1]{#1}

\providetoggle{appendix-pgs}
\settoggle{appendix-pgs}{true}

\providetoggle{appendix-online}
\settoggle{appendix-online}{true}

\algtext*{EndFor}
\algtext*{EndIf}
\algtext*{EndWhile}

%% file: abstract.tex
\begin{abstract}
Recent {\em beyond worst-case optimal} join algorithms Minesweeper and its generalization Tetris have brought  the theory of indexing and join processing together by developing a geometric framework for joins. These algorithms take as input an index $\mathcal{B}$, referred to as a {\em box cover}, that stores {\em output gaps} that can be inferred from traditional indexes, such as B+ trees or tries, on the input relations.
The performances of these algorithms highly depend on the {\em certificate} of $\mathcal{B}$, which is the smallest subset of gaps in $\mathcal{B}$ whose union covers all of the gaps in the output space of a query $Q$. Different box covers can have different size certificates and the sizes of both the box covers and certificates highly depend on the ordering of the domain values of the attributes in $Q$. We study how to generate box covers that contain small size certificates to guarantee efficient runtimes for these algorithms. First, given a query $Q$ over a set of relations of size $N$ and a fixed set of domain orderings for the attributes,  we give a $\tilde{O}(N)$-time algorithm called {\em GAMB} which generates a box cover for $Q$ that is guaranteed to contain the smallest size certificate across any box cover for $Q$.  Second, we show that finding a domain ordering to minimize the box cover size and certificate is NP-hard through a reduction from the {\em 2 consecutive block minimization problem} on boolean matrices. 
 Our third contribution is a $\tilde{O}(N)$-time approximation algorithm called {\em ADORA} to compute domain orderings, under which one can compute a box cover of size $\tilde{O}(K^r)$, where $K$ is the minimum box cover for $Q$ under any domain ordering and $r$ is the maximum arity of any relation. This guarantees certificates of size  $\tilde{O}(K^r)$. We combine ADORA and GAMB with Tetris to form a new algorithm we call {\em TetrisReordered}, which provides
 several new beyond worst-case bounds. On infinite families of queries, TetrisReordered's runtimes are unboundedly better than the bounds stated in prior work.
\end{abstract}

%% file: introduction.tex
\section{Introduction}
\label{sec:introduction}
Performing the natural join of a set of relational tables is a core operation in relational database management systems. After the celebrated result of Atserias, Grohe and Marx~\cite{agm} that provided a tight bound on the maximum (or worst-case) size of natural join queries, now known as the {\em AGM bound}, a new class of {\em worst-case optimal} join algorithms were introduced whose runtimes are asymptotically bounded by the AGM bound. More recently, Ngo et al. and Abo Khamis et al., respectively,  introduced the Minesweeper~\cite{minesweeper} algorithm, and its generalization Tetris~\cite{tetris,tetris-journal}, which adopt a geometric framework for joins and provide {\em beyond worst-case} guarantees that are closer to the highest algorithmic goal of instance optimality. Henceforth, we focus on the Tetris algorithm, the more general of these two algorithms. 

\input{figs/ex-box-cover}

Let $Q$ be a query over $m$ relations $\mathcal{R}$ and $n$ attributes $\mathcal{A}$. Let $N$ be the total number of tuples in $\mathcal{R}$.
Throughout this paper, to match the notation of reference~\cite{tetris}, we use $\tilde{O}$-notation to hide polylogarithmic factors in $N$ as well as the query dependent factors $m$ and $n$.
Unlike traditional join algorithms that operate on input tuples, Tetris takes as input a {\em box cover} $\mathcal{B}=\cup_{R\in\mathcal{R}}\mathcal{B}_R$, where each $\mathcal{B}_R$ is a set of {\em gap boxes} (i.e., tuple-free regions) of the relation $R$ whose union {\em covers} the complement of $R$.
These boxes imply regions in the output space of queries where output tuples cannot exist. Tetris operates on these gaps by performing {\em geometric resolutions}, which generate new gap boxes. 
The runtime of Tetris is bounded by $\tilde{O}\big(\big(C_{\Box}(\mathcal{B})\big)^{w+1}+Z\big)$\footnote{A second upper bound that depends on the number of attributes instead of $w$ is also provided in~\cite{tetris}.} where: 
(i) $C_{\Box}(\mathcal{B})$ is the size of the {\em box certificate} for $\mathcal{B}$, which is the smallest subset of boxes in $\mathcal{B}$ that cover the gaps in the output, i.e., the complement of the output tuples of the join; (ii) $w$ is the {\em treewidth} of $Q$; and (iii) $Z$ is the number of output tuples. 
Figure~\ref{fig:box-cover-certificate} shows an example of this geometric framework on query $R(A, B) \bowtie S(A, C)$. Purple unit boxes indicate input tuples, the boxes in the box cover are shown with rectangles, and the boxes in the certificate are drawn as red rectangles. This Tetris result is analogous to Yannakakis's data-optimal algorithm for acyclic queries and its combination with worst-case optimal join algorithms, which yields results of the form $\tilde{O}(N^{\text{fhtw}} + Z)$, where fhtw is the {\em fractional hypertree width}~\cite{fhtw} and $N$ is the number of tuples in the input. The performance of Tetris's results can be significantly better than Yannakakis-based algorithms, as the certificates are always $\tilde{O}(N)$ and can be $o(N)$, e.g. constant size, on some inputs.

\input{figs/ex-reordering}

In references~\cite{tetris} and~\cite{minesweeper}, a box cover was assumed to be inferred from the available indexes on the relations. Consider a B+ tree index on a relation $R(A, B)$ with sort order ($A$, $B$) and two consecutive tuples ($a_1$, $b_1$) and ($a_1$, $b_2$).\footnote{This example is borrowed from reference~\cite{tetris}.} From these two tuples, a system can infer a gap box ($a_1$, [$b_1+1$, $b_2-1$]) in the output space of any join query that involves $R$. 
The boxes in Figure~\ref{fig:box-cover-certificate} are inferred from B+ tree indexes on $R$ and $S$ with sort orders $(A,B)$ and $(A,C)$, respectively.
\kaleb{Using different indexes can result in box covers with vastly different certificate sizes.} This motivates the first question we study in this paper:

\noindent {\em Question 1: How can a system efficiently generate a good box cover for a set of relations?}

Given a query $Q$, let $C_{\Box}(Q)$
be the minimum certificate size across all possible box covers for the relations in $Q$.\footnote{Note that our use of the notation $C_{\Box}(Q)$ is different from reference~\cite{tetris}, where $\mathcal{B}$  was assumed to be given, and $C_{\Box}(Q)$ was used to indicate the certificate size for $\mathcal{B}$. Since we drop this assumption, $C_{\Box}(\mathcal{B})$ here denotes the certificate size for $\mathcal{B}$ and $C_{\Box}(Q)$ denotes the certificate size over all possible box covers.} An ideal goal for a system would be to efficiently generate a box cover whose certificate is of size $C_{\Box}(Q)$, ensuring performance as a function of $C_{\Box}(Q)$.  We refer to this problem as \texttt{BoxMinC}.
We present a surprisingly positive result for \texttt{BoxMinC}:


\begin{theorem}
\label{thm:box-generation}
Given a database $D$, there is a $\tilde{O}(N)$-time algorithm that can generate a box cover $\mathcal{B}$ of size at most $\tilde{O}(N)$ that 
 contains a certificate of size  $\tilde{O}(C_{\Box}(Q))$  for any join query $Q$ over any subset of relations in $D$. 
\end{theorem}
Therefore, in $\tilde{O}(N)$ time and space, a system can generate a globally good box cover (an index) for all possible join queries over a database.\footnote{Since any system has to spend $\Omega(N)$ time to index its tuples, this time is within an $\tilde{O}(1)$ factor of any other indexing approach. \kaleb{Appendix~\ref{sec:gamb-index-maintenance} shows that a box cover index can also be maintained efficiently. The Tetris runtimes in reference~\cite{tetris} do not add a $\tilde{O}(N)$ indexing component because it is assumed that indexes are given. In practice, this cost must be paid at some point by the database to answer queries.}}
We achieve this result by observing that the set of all maximal gap boxes in the complements of the relations contains a certificate of size $|C_{\Box}(Q)|$ and we provide an $\tilde{O}(N)$-time algorithm called {\em GAMB} that {\bf g}enerates {\bf a}ll {\bf m}aximal dyadic gap {\bf b}oxes (and possibly some non-maximal ones) from the relations. 


In the second question we study, we consider evaluating a single query $Q$. There are simple queries
which can be geometrically complex and require large box covers and certificates.
In many cases, these queries can be modified by reordering each attribute's domain
so that smaller covers and certificates are possible.
Figure~\ref{fig:reorder-rows-cols} shows an example.
In the example, the queries $Q$ and
$Q'$ are both triangle queries joining three binary relations.
These queries are equivalent up to reordering the domains of each attribute.
That is,  it is possible to reorder the rows and columns of the grid
in Figure~\ref{fig:reorder-rows-cols-1} to obtain Figure~\ref{fig:reorder-rows-cols-2}.
Let $\sigma$ be the set of three permutations on the domains of $A,B$, and $C$ which transforms
$Q$ into $Q'$. Specifically, for each attribute, $\sigma$ maps the even values to values between 000 and 011, and the odd values to values between 100 and 111. Despite their equivalence up to reorderings, $Q$ requires a box cover of size 96, as each white grid cell in Figure~\ref{fig:reorder-rows-cols-1} must have a unit gap box covering it, while $Q'$ only requires a box cover size 6. 
The same also applies to the certificate sizes, as every gap box in the box cover must also be part of the box certificate in this case. By extending the domains of the attributes, the difference in box cover and certificate sizes can be made arbitrarily large.
Therefore, a system could improve the performance of Tetris significantly by reordering the domains of attributes.
This motivates our second question:

\noindent {\em Question 2: How can a system efficiently reorder the domains to obtain a small box cover?} 


We refer to the problem of finding a domain ordering $\sigma$ such that the minimum box cover size under $\sigma$ is minimized as \texttt{DomOr}$_{\text{\texttt{BoxMinB}}}$. Let $\mathcal{B}^*$ be the minimum size box cover for a query under any domain ordering, $K=|\mathcal{B}^*|$, and $\sigma^*$ be the ordering under which $\mathcal{B}^*$ is achieved. We first provide a hardness result showing that computing $\sigma^*$ is NP-hard through a reduction from the {\em 2 consecutive block minimization problem} on boolean matrices~\cite{2cbmp-nph}. 
We then provide an approximation algorithm, which we refer to as {\em ADORA}, for \textbf{A}pproximate \textbf{D}omain \textbf{Or}dering \textbf{A}lgorithm, to obtain the following result:

\begin{theorem}
\label{thm:intro-qdfbcmp-approx}
Let $r$ be the maximum arity of any relation in the query $Q$
\kaleb{and let $K$ be the minimum box cover size for $Q$ under any domain ordering.}
There is a $\tilde{O}(N)$-time algorithm that computes a domain ordering $\sigma$ for $Q$, under which one can compute a box cover of size $\tilde{O}(K^r)$, guaranteeing a certificate of size $\tilde{O}(K^r)$.
\end{theorem}
After $\sigma$ is obtained with ADORA, a system can run GAMB, which has the same asymptotic runtime, to obtain a box cover that guarantees certificates of size $\tilde{O}(K^r)$. 
ADORA is based on an intuitive and powerful heuristic that groups the domain values in an attribute that have identical value combinations in the remaining attributes across the relations and makes the values in each group consecutive. Our approximation ratio does not depend on any other parameters of the query, such as different notions of width or the number of relations. 
Once an ordering is obtained, Tetris can be executed on the reordered query and the results converted back to the original domain. This technique is formalized in our algorithm {\em TetrisReordered}.
\kaleb{We construct families of queries for which Tetris on a default ordering has a polynomial runtime with an arbitrarily high degree, but for which TetrisReordered runs in $\tilde{O}(N)$ time.}

%% file: figs/ex-box-cover.tex
\begin{figure}
\vspace{-10pt}
\begin{center}
\begin{subfigure}{0.35\textwidth}
\begin{center}

\begin{tikzpicture}
\begin{axis}[axis lines=left,xlabel=$B$,ylabel=$A$,
  xmin=0,xmax=8,
  ymin=0,ymax=8,
  tickwidth=0,
  xticklabel interval boundaries,yticklabel interval boundaries,
  xtick={0,1,2,3,4,5,6,7,8},ytick={0,1,2,3,4,5,6,7,8},
  xticklabels={000,001,010,011,100,101,110,111},yticklabels={000,001,010,011,100,101,110,111},
  xticklabel style={rotate=90},
  y label style={at={(axis description cs:0.05,0.5)},anchor=north},
  x label style={at={(axis description cs:0.5,-0.2)},anchor=south},
  grid=both,axis line style={-},
  height=4.3cm,width=4.3cm]
\filldraw[\tuplecolour,opacity=\tupleopacity] (axis cs:0,0) rectangle (axis cs:1,1);
\filldraw[\tuplecolour,opacity=\tupleopacity] (axis cs:0,5) rectangle (axis cs:1,6);
\filldraw[\tuplecolour,opacity=\tupleopacity] (axis cs:3,1) rectangle (axis cs:4,2);
\filldraw[\tuplecolour,opacity=\tupleopacity] (axis cs:3,4) rectangle (axis cs:4,5);
\filldraw[\tuplecolour,opacity=\tupleopacity] (axis cs:6,4) rectangle (axis cs:7,5);
\filldraw[\tuplecolour,opacity=\tupleopacity] (axis cs:7,0) rectangle (axis cs:8,1);
\draw[\boxcolour] (axis cs:1.1,0.1) rectangle (axis cs:6.9,0.9);
\draw[\boxcolour] (axis cs:0.1,1.1) rectangle (axis cs:2.9,1.9);
\draw[\boxcolour] (axis cs:4.1,1.1) rectangle (axis cs:7.9,1.9);
\draw[\boxcolour] (axis cs:0.1,4.1) rectangle (axis cs:2.9,4.9);
\draw[\boxcolour] (axis cs:4.1,4.1) rectangle (axis cs:5.9,4.9);
\draw[\boxcolour] (axis cs:7.1,4.1) rectangle (axis cs:7.9,4.9);
\draw[\boxcolour] (axis cs:1.1,5.1) rectangle (axis cs:7.9,5.9);
\draw[\certcolour,thick] (axis cs:0.1,2.1) rectangle (axis cs:7.9,3.9);
\draw[\certcolour,thick] (axis cs:0.1,6.1) rectangle (axis cs:7.9,7.9);
\end{axis}
\end{tikzpicture}
\end{center}
\end{subfigure}
\begin{subfigure}{0.35\textwidth}
\begin{center}

\begin{tikzpicture}
\begin{axis}[axis lines=left,xlabel=$C$,ylabel=$A$,
  xmin=0,xmax=8,
  ymin=0,ymax=8,
  tickwidth=0,
  xticklabel interval boundaries,yticklabel interval boundaries,
  xtick={0,1,2,3,4,5,6,7,8},ytick={0,1,2,3,4,5,6,7,8},
  xticklabels={000,001,010,011,100,101,110,111},yticklabels={000,001,010,011,100,101,110,111},
  xticklabel style={rotate=90},
  y label style={at={(axis description cs:0.05,0.5)},anchor=north},
  x label style={at={(axis description cs:0.5,-0.2)},anchor=south},
  grid=both,axis line style={-},
  height=4.3cm,width=4.3cm]
\filldraw[\tuplecolour,opacity=\tupleopacity] (axis cs:0,2) rectangle (axis cs:1,4);
\filldraw[\tuplecolour,opacity=\tupleopacity] (axis cs:1,2) rectangle (axis cs:2,3);
\filldraw[\tuplecolour,opacity=\tupleopacity] (axis cs:0,6) rectangle (axis cs:1,7);
\filldraw[\tuplecolour,opacity=\tupleopacity] (axis cs:6,2) rectangle (axis cs:7,3);
\filldraw[\tuplecolour,opacity=\tupleopacity] (axis cs:6,6) rectangle (axis cs:7,7);
\draw[\boxcolour] (axis cs:2.1,2.1) rectangle (axis cs:5.9,2.9);
\draw[\boxcolour] (axis cs:7.1,2.1) rectangle (axis cs:7.9,2.9);
\draw[\boxcolour] (axis cs:1.1,3.1) rectangle (axis cs:7.9,3.9);
\draw[\boxcolour] (axis cs:1.1,6.1) rectangle (axis cs:5.9,6.9);
\draw[\boxcolour] (axis cs:7.1,6.1) rectangle (axis cs:7.9,6.9);
\draw[\boxcolour] (axis cs:0.1,7.1) rectangle (axis cs:7.9,7.9);
\draw[\certcolour,thick] (axis cs:0.1,0.1) rectangle (axis cs:7.9,1.9);
\draw[\certcolour,thick] (axis cs:0.1,4.1) rectangle (axis cs:7.9,5.9);
\end{axis}
\end{tikzpicture}
\end{center}
\end{subfigure}
\end{center}
\vspace{-15pt}
\caption{Box cover and certificates of the query $R(A,B)$ $\bowtie$ $S(A,C)$. Red boxes form a box certificate. Red and black boxes together form a box cover.}
\label{fig:box-cover-certificate}
\vspace{-15pt}
\end{figure}
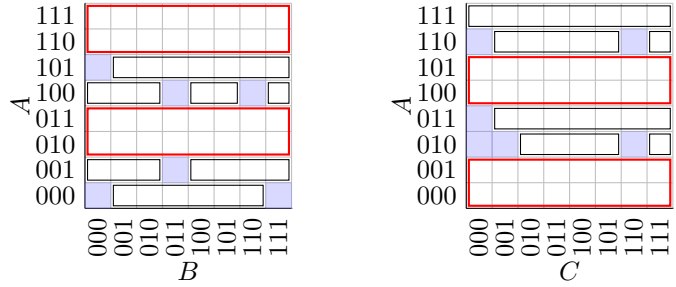

%% file: figs/ex-reordering.tex
\begin{figure*}
\hspace{20pt}
\scalebox{0.9}{\vbox{
\begin{subfigure}{\textwidth}
\captionsetup{justification=centering}
\vspace{-10pt}
\begin{center}
\begin{tabular}{ccc}
$R$ & $S$ & $T$ \\
\begin{tikzpicture}
\begin{axis}[axis lines=left,xlabel=$A$,ylabel=$B$,
  xmin=0,xmax=8,
  ymin=0,ymax=8,
  tickwidth=0,
  xticklabel interval boundaries,yticklabel interval boundaries,
  xtick={0,1,2,3,4,5,6,7,8},ytick={0,1,2,3,4,5,6,7,8},
  xticklabels={000,001,010,011,100,101,110,111},yticklabels={000,001,010,011,100,101,110,111},
  xticklabel style={rotate=90,font=\tiny},
  yticklabel style={font=\tiny},
  y label style={at={(axis description cs:0.1,0.5)},anchor=north},
  x label style={at={(axis description cs:0.5,-0.15)},anchor=south},
  grid=both,axis line style={-},
  height=4.2cm,width=4.2cm]
\filldraw[\tuplecolour,opacity=\tupleopacity] (axis cs:1,0) rectangle (axis cs:2,1);
\filldraw[\tuplecolour,opacity=\tupleopacity] (axis cs:3,0) rectangle (axis cs:4,1);
\filldraw[\tuplecolour,opacity=\tupleopacity] (axis cs:5,0) rectangle (axis cs:6,1);
\filldraw[\tuplecolour,opacity=\tupleopacity] (axis cs:7,0) rectangle (axis cs:8,1);
\filldraw[\tuplecolour,opacity=\tupleopacity] (axis cs:0,1) rectangle (axis cs:1,2);
\filldraw[\tuplecolour,opacity=\tupleopacity] (axis cs:2,1) rectangle (axis cs:3,2);
\filldraw[\tuplecolour,opacity=\tupleopacity] (axis cs:4,1) rectangle (axis cs:5,2);
\filldraw[\tuplecolour,opacity=\tupleopacity] (axis cs:6,1) rectangle (axis cs:7,2);
\filldraw[\tuplecolour,opacity=\tupleopacity] (axis cs:1,2) rectangle (axis cs:2,3);
\filldraw[\tuplecolour,opacity=\tupleopacity] (axis cs:3,2) rectangle (axis cs:4,3);
\filldraw[\tuplecolour,opacity=\tupleopacity] (axis cs:5,2) rectangle (axis cs:6,3);
\filldraw[\tuplecolour,opacity=\tupleopacity] (axis cs:7,2) rectangle (axis cs:8,3);
\filldraw[\tuplecolour,opacity=\tupleopacity] (axis cs:0,3) rectangle (axis cs:1,4);
\filldraw[\tuplecolour,opacity=\tupleopacity] (axis cs:2,3) rectangle (axis cs:3,4);
\filldraw[\tuplecolour,opacity=\tupleopacity] (axis cs:4,3) rectangle (axis cs:5,4);
\filldraw[\tuplecolour,opacity=\tupleopacity] (axis cs:6,3) rectangle (axis cs:7,4);
\filldraw[\tuplecolour,opacity=\tupleopacity] (axis cs:1,4) rectangle (axis cs:2,5);
\filldraw[\tuplecolour,opacity=\tupleopacity] (axis cs:3,4) rectangle (axis cs:4,5);
\filldraw[\tuplecolour,opacity=\tupleopacity] (axis cs:5,4) rectangle (axis cs:6,5);
\filldraw[\tuplecolour,opacity=\tupleopacity] (axis cs:7,4) rectangle (axis cs:8,5);
\filldraw[\tuplecolour,opacity=\tupleopacity] (axis cs:0,5) rectangle (axis cs:1,6);
\filldraw[\tuplecolour,opacity=\tupleopacity] (axis cs:2,5) rectangle (axis cs:3,6);
\filldraw[\tuplecolour,opacity=\tupleopacity] (axis cs:4,5) rectangle (axis cs:5,6);
\filldraw[\tuplecolour,opacity=\tupleopacity] (axis cs:6,5) rectangle (axis cs:7,6);
\filldraw[\tuplecolour,opacity=\tupleopacity] (axis cs:1,6) rectangle (axis cs:2,7);
\filldraw[\tuplecolour,opacity=\tupleopacity] (axis cs:3,6) rectangle (axis cs:4,7);
\filldraw[\tuplecolour,opacity=\tupleopacity] (axis cs:5,6) rectangle (axis cs:6,7);
\filldraw[\tuplecolour,opacity=\tupleopacity] (axis cs:7,6) rectangle (axis cs:8,7);
\filldraw[\tuplecolour,opacity=\tupleopacity] (axis cs:0,7) rectangle (axis cs:1,8);
\filldraw[\tuplecolour,opacity=\tupleopacity] (axis cs:2,7) rectangle (axis cs:3,8);
\filldraw[\tuplecolour,opacity=\tupleopacity] (axis cs:4,7) rectangle (axis cs:5,8);
\filldraw[\tuplecolour,opacity=\tupleopacity] (axis cs:6,7) rectangle (axis cs:7,8);
\end{axis}
\end{tikzpicture}
&
\begin{tikzpicture}
\begin{axis}[axis lines=left,xlabel=$B$,ylabel=$C$,
  xmin=0,xmax=8,
  ymin=0,ymax=8,
  tickwidth=0,
  xticklabel interval boundaries,yticklabel interval boundaries,
  xtick={0,1,2,3,4,5,6,7,8},ytick={0,1,2,3,4,5,6,7,8},
  xticklabels={000,001,010,011,100,101,110,111},yticklabels={000,001,010,011,100,101,110,111},
  xticklabel style={rotate=90,font=\tiny},
  yticklabel style={font=\tiny},
  y label style={at={(axis description cs:0.1,0.5)},anchor=north},
  x label style={at={(axis description cs:0.5,-0.15)},anchor=south},
  grid=both,axis line style={-},
  height=4.2cm,width=4.2cm]
\filldraw[\tuplecolour,opacity=\tupleopacity] (axis cs:1,0) rectangle (axis cs:2,1);
\filldraw[\tuplecolour,opacity=\tupleopacity] (axis cs:3,0) rectangle (axis cs:4,1);
\filldraw[\tuplecolour,opacity=\tupleopacity] (axis cs:5,0) rectangle (axis cs:6,1);
\filldraw[\tuplecolour,opacity=\tupleopacity] (axis cs:7,0) rectangle (axis cs:8,1);
\filldraw[\tuplecolour,opacity=\tupleopacity] (axis cs:0,1) rectangle (axis cs:1,2);
\filldraw[\tuplecolour,opacity=\tupleopacity] (axis cs:2,1) rectangle (axis cs:3,2);
\filldraw[\tuplecolour,opacity=\tupleopacity] (axis cs:4,1) rectangle (axis cs:5,2);
\filldraw[\tuplecolour,opacity=\tupleopacity] (axis cs:6,1) rectangle (axis cs:7,2);
\filldraw[\tuplecolour,opacity=\tupleopacity] (axis cs:1,2) rectangle (axis cs:2,3);
\filldraw[\tuplecolour,opacity=\tupleopacity] (axis cs:3,2) rectangle (axis cs:4,3);
\filldraw[\tuplecolour,opacity=\tupleopacity] (axis cs:5,2) rectangle (axis cs:6,3);
\filldraw[\tuplecolour,opacity=\tupleopacity] (axis cs:7,2) rectangle (axis cs:8,3);
\filldraw[\tuplecolour,opacity=\tupleopacity] (axis cs:0,3) rectangle (axis cs:1,4);
\filldraw[\tuplecolour,opacity=\tupleopacity] (axis cs:2,3) rectangle (axis cs:3,4);
\filldraw[\tuplecolour,opacity=\tupleopacity] (axis cs:4,3) rectangle (axis cs:5,4);
\filldraw[\tuplecolour,opacity=\tupleopacity] (axis cs:6,3) rectangle (axis cs:7,4);
\filldraw[\tuplecolour,opacity=\tupleopacity] (axis cs:1,4) rectangle (axis cs:2,5);
\filldraw[\tuplecolour,opacity=\tupleopacity] (axis cs:3,4) rectangle (axis cs:4,5);
\filldraw[\tuplecolour,opacity=\tupleopacity] (axis cs:5,4) rectangle (axis cs:6,5);
\filldraw[\tuplecolour,opacity=\tupleopacity] (axis cs:7,4) rectangle (axis cs:8,5);
\filldraw[\tuplecolour,opacity=\tupleopacity] (axis cs:0,5) rectangle (axis cs:1,6);
\filldraw[\tuplecolour,opacity=\tupleopacity] (axis cs:2,5) rectangle (axis cs:3,6);
\filldraw[\tuplecolour,opacity=\tupleopacity] (axis cs:4,5) rectangle (axis cs:5,6);
\filldraw[\tuplecolour,opacity=\tupleopacity] (axis cs:6,5) rectangle (axis cs:7,6);
\filldraw[\tuplecolour,opacity=\tupleopacity] (axis cs:1,6) rectangle (axis cs:2,7);
\filldraw[\tuplecolour,opacity=\tupleopacity] (axis cs:3,6) rectangle (axis cs:4,7);
\filldraw[\tuplecolour,opacity=\tupleopacity] (axis cs:5,6) rectangle (axis cs:6,7);
\filldraw[\tuplecolour,opacity=\tupleopacity] (axis cs:7,6) rectangle (axis cs:8,7);
\filldraw[\tuplecolour,opacity=\tupleopacity] (axis cs:0,7) rectangle (axis cs:1,8);
\filldraw[\tuplecolour,opacity=\tupleopacity] (axis cs:2,7) rectangle (axis cs:3,8);
\filldraw[\tuplecolour,opacity=\tupleopacity] (axis cs:4,7) rectangle (axis cs:5,8);
\filldraw[\tuplecolour,opacity=\tupleopacity] (axis cs:6,7) rectangle (axis cs:7,8);
\end{axis}
\end{tikzpicture}
&
\begin{tikzpicture}
\begin{axis}[axis lines=left,xlabel=$A$,ylabel=$C$,
  xmin=0,xmax=8,
  ymin=0,ymax=8,
  tickwidth=0,
  xticklabel interval boundaries,yticklabel interval boundaries,
  xtick={0,1,2,3,4,5,6,7,8},ytick={0,1,2,3,4,5,6,7,8},
  xticklabels={000,001,010,011,100,101,110,111},yticklabels={000,001,010,011,100,101,110,111},
  xticklabel style={rotate=90,font=\tiny},
  yticklabel style={font=\tiny},
  y label style={at={(axis description cs:0.1,0.5)},anchor=north},
  x label style={at={(axis description cs:0.5,-0.15)},anchor=south},
  grid=both,axis line style={-},
  height=4.2cm,width=4.2cm]
\filldraw[\tuplecolour,opacity=\tupleopacity] (axis cs:1,0) rectangle (axis cs:2,1);
\filldraw[\tuplecolour,opacity=\tupleopacity] (axis cs:3,0) rectangle (axis cs:4,1);
\filldraw[\tuplecolour,opacity=\tupleopacity] (axis cs:5,0) rectangle (axis cs:6,1);
\filldraw[\tuplecolour,opacity=\tupleopacity] (axis cs:7,0) rectangle (axis cs:8,1);
\filldraw[\tuplecolour,opacity=\tupleopacity] (axis cs:0,1) rectangle (axis cs:1,2);
\filldraw[\tuplecolour,opacity=\tupleopacity] (axis cs:2,1) rectangle (axis cs:3,2);
\filldraw[\tuplecolour,opacity=\tupleopacity] (axis cs:4,1) rectangle (axis cs:5,2);
\filldraw[\tuplecolour,opacity=\tupleopacity] (axis cs:6,1) rectangle (axis cs:7,2);
\filldraw[\tuplecolour,opacity=\tupleopacity] (axis cs:1,2) rectangle (axis cs:2,3);
\filldraw[\tuplecolour,opacity=\tupleopacity] (axis cs:3,2) rectangle (axis cs:4,3);
\filldraw[\tuplecolour,opacity=\tupleopacity] (axis cs:5,2) rectangle (axis cs:6,3);
\filldraw[\tuplecolour,opacity=\tupleopacity] (axis cs:7,2) rectangle (axis cs:8,3);
\filldraw[\tuplecolour,opacity=\tupleopacity] (axis cs:0,3) rectangle (axis cs:1,4);
\filldraw[\tuplecolour,opacity=\tupleopacity] (axis cs:2,3) rectangle (axis cs:3,4);
\filldraw[\tuplecolour,opacity=\tupleopacity] (axis cs:4,3) rectangle (axis cs:5,4);
\filldraw[\tuplecolour,opacity=\tupleopacity] (axis cs:6,3) rectangle (axis cs:7,4);
\filldraw[\tuplecolour,opacity=\tupleopacity] (axis cs:1,4) rectangle (axis cs:2,5);
\filldraw[\tuplecolour,opacity=\tupleopacity] (axis cs:3,4) rectangle (axis cs:4,5);
\filldraw[\tuplecolour,opacity=\tupleopacity] (axis cs:5,4) rectangle (axis cs:6,5);
\filldraw[\tuplecolour,opacity=\tupleopacity] (axis cs:7,4) rectangle (axis cs:8,5);
\filldraw[\tuplecolour,opacity=\tupleopacity] (axis cs:0,5) rectangle (axis cs:1,6);
\filldraw[\tuplecolour,opacity=\tupleopacity] (axis cs:2,5) rectangle (axis cs:3,6);
\filldraw[\tuplecolour,opacity=\tupleopacity] (axis cs:4,5) rectangle (axis cs:5,6);
\filldraw[\tuplecolour,opacity=\tupleopacity] (axis cs:6,5) rectangle (axis cs:7,6);
\filldraw[\tuplecolour,opacity=\tupleopacity] (axis cs:1,6) rectangle (axis cs:2,7);
\filldraw[\tuplecolour,opacity=\tupleopacity] (axis cs:3,6) rectangle (axis cs:4,7);
\filldraw[\tuplecolour,opacity=\tupleopacity] (axis cs:5,6) rectangle (axis cs:6,7);
\filldraw[\tuplecolour,opacity=\tupleopacity] (axis cs:7,6) rectangle (axis cs:8,7);
\filldraw[\tuplecolour,opacity=\tupleopacity] (axis cs:0,7) rectangle (axis cs:1,8);
\filldraw[\tuplecolour,opacity=\tupleopacity] (axis cs:2,7) rectangle (axis cs:3,8);
\filldraw[\tuplecolour,opacity=\tupleopacity] (axis cs:4,7) rectangle (axis cs:5,8);
\filldraw[\tuplecolour,opacity=\tupleopacity] (axis cs:6,7) rectangle (axis cs:7,8);
\end{axis}
\end{tikzpicture}
\end{tabular}
\end{center}
\vspace{-15pt}
\caption{$Q=R\Join S\Join T$}
\label{fig:reorder-rows-cols-1}
\end{subfigure}

\vspace{5pt}
\hspace{-7pt}
\begin{subfigure}{\textwidth}
\captionsetup{justification=centering}
\begin{center}
\begin{tabular}{ccc}
$R'$ & $S'$ & $T'$ \\
\begin{tikzpicture}
\begin{axis}[axis lines=left,xlabel={$A,\sigma(A)$},ylabel={$B,\sigma(B)$},
  xmin=0,xmax=8,
  ymin=0,ymax=8,
  tickwidth=0,
  xticklabel interval boundaries,yticklabel interval boundaries,
  xtick={0,1,2,3,4,5,6,7,8},ytick={0,1,2,3,4,5,6,7,8},
  xticklabels={000 000,010 001,100 010,110 011,001 100,011 101,101 110,111 111},yticklabels={000 000,010 001,100 010,110 011,001 100,011 101,101 110,111 111},
  xticklabel style={rotate=90,font=\tiny},
  yticklabel style={font=\tiny},
  y label style={at={(axis description cs:-0.1,0.5)},anchor=north},
  x label style={at={(axis description cs:0.5,-0.4)},anchor=south},
  grid=both,axis line style={-},
  height=4.2cm,width=4.2cm]
\filldraw[\tuplecolour,opacity=\tupleopacity] (axis cs:0,4) rectangle (axis cs:4,8);
\filldraw[\tuplecolour,opacity=\tupleopacity] (axis cs:4,0) rectangle (axis cs:8,4);
\end{axis}
\end{tikzpicture}
&
\begin{tikzpicture}
\begin{axis}[axis lines=left,xlabel={$B,\sigma(B)$},ylabel={$C,\sigma(C)$},
  xmin=0,xmax=8,
  ymin=0,ymax=8,
  tickwidth=0,
  xticklabel interval boundaries,yticklabel interval boundaries,
  xtick={0,1,2,3,4,5,6,7,8},ytick={0,1,2,3,4,5,6,7,8},
  xticklabels={000 000,010 001,100 010,110 011,001 100,011 101,101 110,111 111},yticklabels={000 000,010 001,100 010,110 011,001 100,011 101,101 110,111 111},
  xticklabel style={rotate=90,font=\tiny},
  yticklabel style={font=\tiny},
  y label style={at={(axis description cs:-0.1,0.5)},anchor=north},
  x label style={at={(axis description cs:0.5,-0.4)},anchor=south},
  grid=both,axis line style={-},
  height=4.2cm,width=4.2cm]
\filldraw[\tuplecolour,opacity=\tupleopacity] (axis cs:0,4) rectangle (axis cs:4,8);
\filldraw[\tuplecolour,opacity=\tupleopacity] (axis cs:4,0) rectangle (axis cs:8,4);
\end{axis}
\end{tikzpicture}
&
\begin{tikzpicture}
\begin{axis}[axis lines=left,xlabel={$A,\sigma(A)$},ylabel={$C,\sigma(C)$},
  xmin=0,xmax=8,
  ymin=0,ymax=8,
  tickwidth=0,
  xticklabel interval boundaries,yticklabel interval boundaries,
  xtick={0,1,2,3,4,5,6,7,8},ytick={0,1,2,3,4,5,6,7,8},
  xticklabels={000 000,010 001,100 010,110 011,001 100,011 101,101 110,111 111},yticklabels={000 000,010 001,100 010,110 011,001 100,011 101,101 110,111 111},
  xticklabel style={rotate=90,font=\tiny},
  yticklabel style={font=\tiny},
  y label style={at={(axis description cs:-0.1,0.5)},anchor=north},
  x label style={at={(axis description cs:0.5,-0.4)},anchor=south},
  grid=both,axis line style={-},
  height=4.2cm,width=4.2cm]
\filldraw[\tuplecolour,opacity=\tupleopacity] (axis cs:0,4) rectangle (axis cs:4,8);
\filldraw[\tuplecolour,opacity=\tupleopacity] (axis cs:4,0) rectangle (axis cs:8,4);
\end{axis}
\end{tikzpicture}
\end{tabular}
\end{center}
\vspace{-15pt}
\caption{$Q'=R'\Join S'\Join T'=\sigma(R)\Join\sigma(S)\Join\sigma(T)$}
\label{fig:reorder-rows-cols-2}
\end{subfigure}
}}
\vspace{-5pt}
\caption{Two equivalent queries (up to attribute reorderings) with different box certificate sizes.}
\label{fig:reorder-rows-cols}
\vspace{-15pt}
\end{figure*}
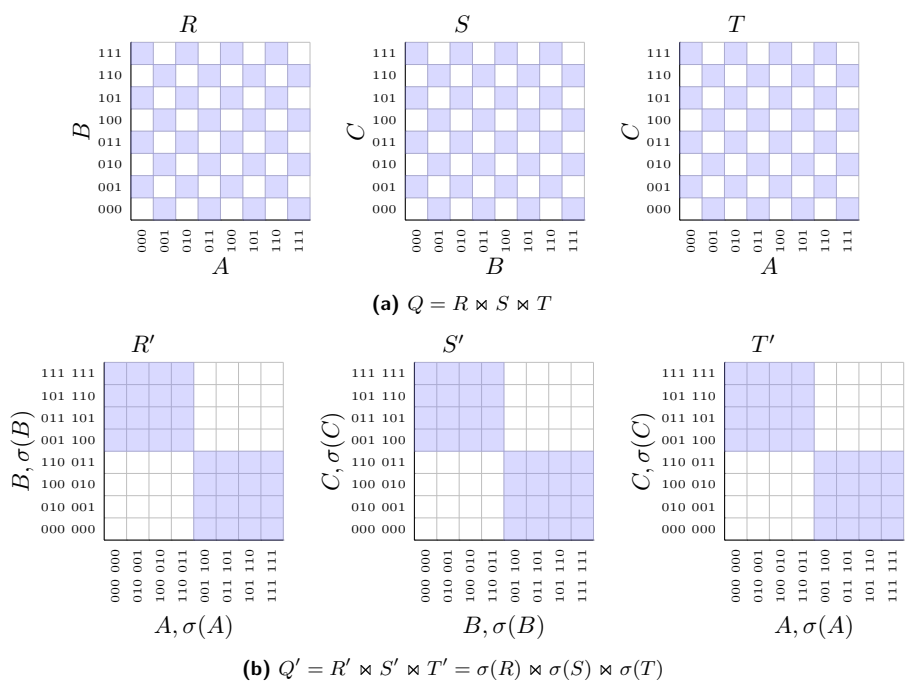

%% file: preliminaries.tex
\vspace{-10pt}
\section{Notation and Preliminaries}
\label{sec:preliminaries}

\kaleb{Throughout this paper, we work with a fixed database $D$.
A \textit{query} $Q$ is an
equi-join over a set of $m$ fixed relations $\mathcal{R}$ and a set of $n$ attributes $\mathcal{A}$ from $D$. We do not differentiate between a query and a query instance, so $Q$ refers to the instance of $Q$ in $D$.}
As in reference~\cite{tetris}, for ease of presentation we assume the domains of each attribute $A\in\mathcal{A}$ consist of all $d$ bit integers
but our results only require domain values to be discrete and ordered.
For $R\in\mathcal{R}$ and $A\in\mathcal{A}$, the attribute set of $R$ is denoted $\attr(R)$ and the domain of $A$ is denoted $\dom(A)$.

Tetris takes as input a box cover $\mathcal{B}$ that contains {\em dyadic} gap boxes, which are boxes whose span over each attribute is encoded as a binary prefix. Let $R\in\mathcal{R}$  contain $n_R$ attributes. Formally, a dyadic gap box in $\mathcal{B}_R$ is an $n_R$-tuple $b=\langle s_1,s_2,\ldots,s_{n_R}\rangle$ where each $s_i$ is a binary string of length at most $d$. We use $*$ to denote the empty string.
\kaleb{We sometimes use $b.A$ to denote the prefix in $b$ corresponding to attribute $A$.}
For example, if $d$ is 3, the dyadic box $\langle 01, 1\rangle$ for $R$($A_1$, $A_2$) is the box whose $A_1$ and $A_2$ dimensions include all values with prefix $01$ and $1$, respectively, i.e., it is the rectangle with sides $\langle [010-011], [100-111]\rangle$. 
Using dyadic boxes allows Tetris to perform {\em geometric resolutions} (explained momentarily) efficiently, which is needed to prove the runtime bounds of Tetris. 
%

Although the details of how Tetris works are not necessary to understand our techniques and contributions, we give a brief overview as background and refer the reader to reference~\cite{tetris} for details. Assume each box in $\mathcal{B}$, say those coming from $\mathcal{B}_R$, are extended, with prefix $*$, to every attribute not in $\attr(R)$. This allows us to think of $\mathcal{B}$ as a single gap box index over the output space. The core of Tetris is a recursive subroutine that determines whether the set of boxes in  $\mathcal{B}$ covers the entire $n$-dimensional output space $\langle *,*,\ldots,*\rangle$ and returns either YES or NO with an output tuple $o$ as a witness. The witnesses are inserted into $\mathcal{B}$. During the execution, this subroutine performs {\em geometric resolutions} that take two boxes that are adjacent in one dimension and construct a new box that consists of the union of the intervals in this dimension (and the intersection in all others). When boxes are dyadic, geometric resolution can be done in $\tilde{O}(1)$ time.
This recursive subroutine is called as many times as there are output tuples until it finally returns YES.
Two variants of Tetris, called Tetris-Preloaded and Tetris-LoadBalanced run in time $\tilde{O}(C_{\Box}(\mathcal{B})^{w+1}+ Z)$ and $\tilde{O}(C_{\Box}(\mathcal{B})^{n/2} + Z)$, respectively (see Theorems 4.9 and 4.11 in reference~\cite{tetris}). 
$C_{\Box}(\mathcal{B})$ in Tetris's runtime is the box certificate size of $\mathcal{B}$, which is the size of the smallest subset $\mathcal{B'}$ of  $\mathcal{B}$, such that the union of boxes in $\mathcal{B'}$ and the union of boxes in $\mathcal{B}$ cover exactly the same space. Equivalently, $C_{\Box}(\mathcal{B})$ is the size of the smallest subset $\mathcal{B'}$ of $\mathcal{B}$ whose extended boxes (with $*$'s as described above) cover all of the gaps in the output space.

We end this section with a note on dyadic vs. general boxes. The notions of certificate, box cover, and the problems we study can be defined in terms of dyadic or general boxes. Except in Section~\ref{sec:generating-boxes}, the term box refers to general boxes, and our optimization problems are defined over general box covers and certificates.
For both certificates and box covers, the minimum size obtained with dyadic boxes and general boxes are within $\tilde{O}(1)$ of each other.
This is because a dyadic box is a general box by definition and any general box can be partitioned into $\tilde{O}(1)$ dyadic boxes (Proposition B.14 in reference~\cite{tetris}). Our approximation results for general boxes imply approximation results for dyadic boxes up to $\tilde{O}(1)$ factors. However, a hardness result for one version does not imply hardness of the other. Our hardness results apply only to general boxes.
\kaleb{However, we use dyadic boxes extensively because they are a powerful analytical tool which the results of this paper and reference~\cite{tetris} rely on.}

%% file: related.tex
\vspace{-10pt}
\section{Related Work}
\label{sec:related}
 
\subsection{Box Cover Problems}
\label{sec:box-cover-problems}
The complement of a relation $R$ with $k$ attributes can be represented geometrically as a set of axis-aligned, rectilinear polytopes in $k$-dimensional space, which may have holes (the tuples in R form the exteriors of the polytopes).  The number of vertices in these polytopes is bounded (up to a constant factor) by the number of tuples in the relation.
Therefore our work is closely related to covering rectilinear polytopes with a minimum number of rectangles in geometry.  This problem has been previously studied in the 2-dimensional setting, i.e., for polygons. The problem is known to be NP-complete, even when the polygon is hole-free \cite{culberson-reckhow} and MaxSNP-hard for polygons with holes~\cite{berman-dasgupta}. There are several approximation algorithms for the problem. Franzblau \cite{franzblau} designed an algorithm that approximates the optimal solution to a factor of $O(\log n)$, where $n$ is the number of vertices in the polygon. If the polygon is hole-free, the approximation factor improves to~2. Anil Kumar and Ramesh \cite{rectilinear-polygons} showed a tighter approximation ratio of $O(\sqrt{\log n})$ for the same algorithm on polygons with holes. Franzblau et al.~\cite{franzblau-kleitman} also showed the problem is solvable in polynomial time in the special case when polygons are vertically convex. All of these results are limited to 2D and little is known about the problem in higher dimensions.

The approximation algorithms above can be used to generate box covers for the complement of a binary relation $R$. This is a special case of \texttt{BoxMinC}, where the input is a trivial query with a single binary relation $R$. Outside of this limited setting, the connection of covering axis-aligned and rectilinear polygons to \texttt{BoxMinC} breaks. This is because the certificate of a query in this case is the smallest number of boxes that cover the complement of the output, using boxes from the relations. In this case, because the output is not yet computed, it is not known a priori which polytopes should be covered.


There are variants of covering polygons that are less directly related to our problems. Reference \cite{gudmundsson-levcopoulos} studies the more general problem of covering polygons with only obtuse interior angles, and provides approximation algorithms. Reference~\cite{levcopoulos-gudmundsson-square} studies covering the input polygon with squares instead of rectangles. For a survey of geometric covering and packing problems, including shapes beyond polytopes, we refer the reader to references~\cite{sphere-packings} and~\cite{survey-packing-covering}.

\subsection{Orderings in Matrices}
There are several problems related to ordering the rows and columns of boolean matrices to achieve different optimization goals. The closest to our work is the consecutive block minimization problem (CBMP)~\cite{cbmp-npc}. Our hardness results are based on a variant of CBMP, called 2 consecutive block minimization~\cite{2cbmp-nph}, which we review in Section~\ref{sec:nph}. There are two other ordering problems for boolean matrices, which are less related to our work: (i) the {\em consecutive ones property test} determines whether there is a column ordering such that each row has only one consecutive block of ones~\cite{cop-poly}; (ii)  the {\em doubly lexical ordering problem} finds a row and column ordering such that both rows and columns are in lexicographic order~\cite{doubly-lexical}. Both problems have polynomial time solutions.	


\subsection{Worst-Case and Beyond Worst-Case Join Algorithms}
\label{sec:join-algorithms}
A join algorithm is said to be \emph{worst-case optimal} if it runs in time $\tilde{O}(\AGM(Q))$, where the AGM bound~\cite{agm} is the worst-case upper bound on the number of output tuples for a query based on its shape and the number of input tuples. Examples of worst-case optimal join algorithms are Leapfrog Triejoin~\cite{leapfrog-triejoin}, NPRR~\cite{nprr}, and Generic Join~\cite{generic-join}. 
A survey on worst-case optimal join algorithms can be found in reference~\cite{worst-case-survey}.
There are several results that consider other properties of the query and provide worst-case  upper bounds on the size of query outputs that are better than the AGM bound. 
Olteanu and Z\'{a}vodn\'{y} \cite{factorized} show that worst-case sizes of queries in {\em factorized representations} can be asymptotically smaller than the AGM bound and provide algorithms that meet these factorized bounds. Joglekar and R\'{e} \cite{degree-based} developed an algorithm which provides degree-based worst-case results that assume knowledge of degree information for the values in the query. Similarly, references \cite{bounds-degree-constraints} and \cite{bounds-functional-dependencies} provide worst-case bounds based on information theoretical bounds that take into account, respectively, more general degree constraints and functional dependencies.

Several results go beyond worst-case bounds and are closer to the notion of instance optimality. The earliest example is Yannakakis' data-optimal algorithm \cite{yannakakis} for acyclic queries
that runs in time $O(N+Z)$. This was later
generalized to an algorithm~\cite{tree-clustering-schemes} for arbitrary queries which runs in time $\tilde{O}(N^{fhtw}+Z)$,
where $fhtw$ is the query's {\em fractional hypertree width}~\cite{fhtw}. The Minesweeper algorithm~\cite{minesweeper} developed the measure of {\em comparison certificate} $C_{comp}$  for comparison-based join algorithms, which captures the minimum number of comparisons needed to prove the output of a join query is correct. Minesweeper runs in time
$\tilde{O}(|C_{comp}|^{w+1}+Z)$, where $Z$ is the number of output tuples and $w$ is the query's treewidth. The Tetris algorithm~\cite{tetris}, which motivates our work, generalizes comparison certificates to the geometric notion of a box certificate, reviewed in Section~\ref{sec:introduction}. For every comparison certificate $C_{comp}$, there is a box certificate of size at most $|C_{comp}|$. In this sense,
box certificates are stronger than comparison certificates, and Tetris subsumes the certificate-based results of Minesweeper. Our results on finding box covers with small certificates and domain orderings with small box covers improve the bounds provided by Tetris.

%% file: generating-boxes.tex
\vspace{-10pt}
\section{Generating a Box Cover}
\label{sec:generating-boxes}

Since the runtime of Tetris depends on the certificate size of its input box cover, an important preprocessing step for the algorithm is to generate a box cover with a small certificate. Ideally, a system should generate a box cover that contains a certificate of minimum size, across all box covers. We defined this quantity as $C_{\Box}(Q)$ in Section~\ref{sec:introduction}. \kaleb{The following lemma states two facts about dyadic boxes that are crucial for our results and the results in reference~\cite{tetris}.}
\begin{lemma}
\label{lemma:dyadic}
(Propositions B.12 and B.14 \cite{tetris})
 Let $b$ be any dyadic box. Then there are $\tilde{O}(1)$ dyadic boxes which contain $b$.
Let $b'$ be any (not necessarily dyadic) box. Then $b'$ can be partitioned into a set of $\tilde{O}(1)$
disjoint dyadic boxes whose union is equal to $b'$.
\end{lemma}

Let a dyadic gap box $b$ for a relation $R$ be \emph{maximal} if $b$ cannot be enlarged in any of its dimensions and still remain a dyadic gap box, i.e., not include an input tuple of $R$. Generating a box cover with certificate size $\tilde{O}(C_{\Box}(Q))$ can be done by generating the set of all maximal dyadic gap boxes in the input relations. This is because: (1) any general box
can be decomposed into $\tilde{O}(1)$ dyadic boxes by Lemma~\ref{lemma:dyadic}, so decomposing a general box cover into a dyadic one can increase its certificate size by at most a factor of $\tilde{O}(1)$; and (2) expanding any non-maximal dyadic boxes to make them maximal can only decrease the size of the certificate.
We will show that given any query $Q$ with $N$ input tuples, we can generate all maximal dyadic gap boxes over all of the relations in $Q$ in $\tilde{O}(N)$ time. This also implies that the number of maximal dyadic boxes is $\tilde{O}(N)$.  Interestingly, this is not true for general gap boxes, of which there can be a super-linear number (see Appendix~\ref{sec:many-general-boxes} for an example).

Algorithm~\ref{alg:gamb} shows the pseudocode for our algorithm {\em GAMB} that
\textbf{g}enerates \textbf{a}ll \textbf{m}aximal dyadic gap \textbf{b}oxes for a relation $R$ in $\tilde{O}(N)$ time. GAMB loops over each dyadic box $b$ covering each tuple $t$ in $R$, explores boxes that are adjacent to $b$ (which may or may not be gap boxes) and inserts these into a set $B$. Then it subtracts the set of all dyadic boxes covering any tuples from $B$ to obtain a set of gap boxes. As we argue, this set contains every maximal dyadic gap box (and possibly some non-maximal ones).
To generate all maximal boxes for a query $Q=(\mathcal{R},\mathcal{A})$, we can simply iterate over each $R \in\mathcal{R}$ and invoke GAMB. 	

\begin{algorithm}[t]

\caption{GAMB($R$): Generates all maximal dyadic gap boxes of $R$.}
\label{alg:gamb}
\begin{algorithmic}[1]
\State $B:=\emptyset$, $\bar{B}:=\emptyset$
\For{$t\in R$}
  \For{every dyadic box $b$ such that $t\in b$}\label{gamb:second:loop}
  \State $\bar{B}:= \bar{B}\cup\{b\}$
    \For{$A\in\attr(R)$ such that $b.A\neq *$} \label{gamb:third:loop}
      \State Let $b'$ be the box when the last bit of $b.A$ is flipped  \label{gamb:third:loop-inside}
      \State $B:= B\cup\{b'\}$
    \EndFor
  \EndFor
\EndFor
\State\Return $B\setminus\bar{B}$  \label{gamb:set-difference}
\end{algorithmic}
\end{algorithm}

\begin{theorem}
\label{thm:gamb}
GAMB generates all maximal dyadic gap boxes of a relation $R$ in $\tilde{O}(N)$ time.
\end{theorem}

\begin{proof}

Let $b'$ be a maximal dyadic gap box for $R$. Let $A$ be an attribute of $R$ for which $b'$
specifies at least one bit (so $b'.A\neq *$). Let $b$ be the dyadic box obtained from $b'$ by flipping the last
bit of $b'.A$. Since $b'$ is maximal, $b$ contains at least one tuple $t\in R$.
Since $b$ is a dyadic box containing $t$, some iteration of the for-loop on line~\ref{gamb:second:loop} will
reach box $b$. Then the for-loop on line~\ref{gamb:third:loop} at some iteration will loop over $A$ and generate exactly $b'$ on line~\ref{gamb:third:loop-inside}. Thus $b'$ is added to $B$ and since $b'$ is a gap box, GAMB will not add it to $\bar{B}$ (which only contains non-gap boxes). Therefore $b'$ will be in the output of GAMB. Note that the returned set does not contain any non-gap boxes of $R$, since every box which contains any tuple of $R$ is added to $\bar{B}$. 
The outer-most for loop has $N$ iterations. The for loop on line~\ref{gamb:second:loop} has  $\tilde{O}(1)$ iterations by Lemma~\ref{lemma:dyadic}.
The for-loop on line~\ref{gamb:third:loop} has $n$, so $\tilde{O}(1)$, iterations. 
Finally, the set difference on line~\ref{gamb:set-difference} can be done by sorting both $B$ and $\bar{B}$ and iterating lockstep through the sorted boxes. Therefore, the total runtime of GAMB is $\tilde{O}(N)$.
\end{proof}

By our earlier observation based on Lemma~\ref{lemma:dyadic}, running GAMB
as a preprocessing step is sufficient to generate a box cover with a certificate of size $\tilde{O}(C_{\Box}(Q))$. Combined with runtime upper bounds of Tetris from reference~\cite{tetris}, we can state the following corollary:

\begin{corollary}
\label{cor:gamb}
Given a database $D$ of relations with $N$ total tuples, in $\tilde{O}(N)$ preprocessing time, one can generate a box cover $\mathcal{B}$ such that running Tetris on $\mathcal{B}$ yields $\tilde{O}\big(\big(C_{\Box}(Q)\big)^{w+1}+Z\big)$ or
$\tilde{O}\big(\big(C_{\Box}(Q)\big)^{n/2}+Z\big)$ runtimes for any query $Q$ over $D$.
\end{corollary}

One interpretation of this result is that in $\tilde{O}(N)$ preprocessing time, a system can generate {\em a single global index} that will make Tetris efficient on all possible join queries over a database $D$. In fact, using the bounds from reference~\cite{tetris}, these are the best bounds we can obtain up to an $\tilde{O}(1)$ factor when $Q$ is fixed, since $C_{\Box}(Q)$ is the minimum certificate size for any box cover of $Q$.
This is surprisingly achieved with the same index for all queries, so the $\tilde{O}(N)$ preprocessing cost need only be incurred once for a workload of any number of joins.
To improve on these bounds, we must modify $Q$ to reduce the 
box certificate size. We next explore domain orderings as a method to improve these bounds.

%% file: domain-ordering.tex
\section{Domain Ordering Problems}
\label{sec:domain-ordering}

We next study the $\DomOrBox$ problem. Given a query $Q$, our goal is to find the minimum size box cover possible under any {\em domain ordering} for $Q$ and to find the domain ordering $\sigma^*$ that yields this minimum possible box cover size. We begin by defining a domain ordering.

\begin{definition}[Domain ordering]
A \emph{domain
ordering} for a query $Q=(\mathcal{R},\mathcal{A})$ is a tuple of $|\mathcal{A}|$ permutations $\sigma=(\sigma_A)_{A\in\mathcal{A}}$
where each $\sigma_A$ is a permutation of $\dom(A)$.
\end{definition}

\begin{example}
Let $A$ and $B$ be attributes over 2-bit domains. Let $R(A, B)$ be the following relation presented under the default domain ordering $[00,01,10,11]$ for both $A$ and $B$:
\begin{center}
$R(A,B) = \big\{\langle 00,00\rangle, \langle 01,11\rangle, \langle 10,00\rangle, \langle 11, 11\rangle\big\}$
\end{center}
Consider the domain ordering $\sigma$ where $\sigma_A$$=$$\sigma_B$$=$$\{$$00  \mapsto 00$, $01 \mapsto 10$, $10 \mapsto 11$, $11 \mapsto 01$$\}$.
We write $\sigma$ as $\sigma_A=\sigma_B=[00,11,01,10]$ to indicate the new ``locations'' of the previous domain values in the new ordering.
Then $\sigma(R)$
denotes the following relation:
\begin{center}
$\sigma(R)(A,B) = \big\{\langle 00,00\rangle,\langle 10,01\rangle,\langle 11,00\rangle,\langle 01,01\rangle\big\}$
\end{center}
\end{example}

The choice of domain ordering can have a significant effect on box cover sizes and their certificates. We show in Section~\ref{sec:tetris-reordered} that \kaleb{given a query $Q$ over $n$ attributes,
we can construct an infinite family of queries over $n$ attributes which require arbitrarily large
box covers and certificates under a default domain ordering, but under another domain ordering, have box covers and certificates of the same size as $Q$.}
Our specific problem is this: 

\begin{definition}[$\DomOrBox$]
Let $K_{\Box}(\sigma(Q))$ be the minimum box cover size one can obtain for the query $\sigma(Q)$ obtained
from $Q$ by ordering the domains according to $\sigma$. Given a query $Q$, output
a domain ordering $\sigma^*$ such that
$K_{\Box}(\sigma^*(Q)) = \min_{\sigma}K_{\Box}(\sigma(Q))$.
\end{definition}


In Section~\ref{sec:nph}, we show that $\DomOrBox$ is NP-hard.
In Section~\ref{sec:domain-ordering-approx}, we present ADORA, an approximation algorithm for
$\DomOrBox$. Section~\ref{sec:tetris-reordered} combines ADORA, GAMB, and Tetris
in an algorithm we call {\em TetrisReordered}, which has new beyond worst-case bounds.
In Section~\ref{sec:tetris-reordered} we also present infinite classes of queries for
which TetrisReordered
runs unboundedly faster than the versions of Tetris
from reference~\cite{tetris}. 

\subsection{\texorpdfstring{$\DomOrBox$}{DomOrBoxMinB} is NP-hard}
\label{sec:nph}

\input{nph}
\subsection{Approximating \texorpdfstring{$\DomOrBox$}{DomOrBoxMinB}}
\label{sec:domain-ordering-approx}

\input{approx-dom-or}

%% file: nph.tex
Our reduction is from the \emph{2 consecutive block minimization problem (2CBMP)}
on boolean matrices~\cite{2cbmp-nph}.\footnote{In this section, we use the informal convention of discussing the NP-hardness of minimization problems. Since NP-hardness is  defined for decision problems, when we state that a minimization problem like 2CBMP is NP-hard, we are implicitly referring to the decision problem which takes as input an additional positive integer $k$ and accepts if and only if the minimum value of the objective function is at most $k$.}
In a boolean matrix $M$, a \emph{consecutive block} is a maximal consecutive run of 1-cells in a single row of $M$,
which is bounded on the left by either the beginning of the row or a 0-cell, and bounded on the right by
either the end of the row or a 0-cell.
We use $\cb(M)$ to denote the total number of consecutive blocks in $M$ over all rows. 
Let $M$ be a boolean matrix 
stored as a 2D dense array, each row of which contains at most 2 1-cells. 2CBMP is the problem of finding an ordering $\sigma_c^*$ on the columns of $M$ such that
$\cb(\sigma_c^*(M)) = \min_{\sigma_c}\cb(\sigma_c(M))$.
2CBMP was shown to be NP-hard in reference~\cite{2cbmp-nph}.


\begin{theorem}
\label{thm:nph}
$\DomOrBox$ is NP-hard.
\end{theorem}

\begin{proof}
We focus on the special case where $Q$ contains a single relation $R(A,B)$ over exactly 2 attributes and show that $\DomOrBox$ is NP-hard even in this case. 
 This implies $\DomOrBox$ is NP-hard for any number of attributes and relations, since one can duplicate $R$ to another relation $S$ with the same schema,
and extend $R$ and $S$ to a third attribute $C$, taking $R'=R\times\dom(C)$. The ordering that solves $\DomOrBox$ on $R'\bowtie S'$ (a trivial intersection query) also minimizes the box cover size for  $R$.
For the purposes of the proof, we model $R$ as a boolean matrix $M'$, with a row for each value in $\dom(B)$
and a column for each value in $\dom(A)$.
Each cell of the matrix corresponds to a possible tuple in $\dom(A)\times\dom(B)$. The matrix
$M'$ contains a 0-cell in column $i$ and row $j$ if the tuple $t=\langle i,j\rangle\in R$, and a 1-cell otherwise.
This means that a box cover $\mathcal{B}$ for $R$ corresponds directly to a set of rectangles which cover all of the 1-cells
of $M'$, and vice-versa. 
Readers can assume $M'$ is given to  $\DomOrBox$ as a  dense matrix or a list of tuples, i.e., ($i$, $j$) indices for the 0 cells. 

Let $M$ be an $n\times m$ boolean matrix input to 2CBMP.
We construct a $(4n)\times(m+2n)$ matrix $M'$ for input to
$\DomOrBox$. 
For each row $r_i$ of $M$, we create 4 rows in $M'$: $r_{i,1},r_{i,2},p_{i,1}$, and $p_{i,2}$.
 $r_{i,1}$ and $r_{i,2}$ are duplicates of the original row $r_i$, and $p_{i,1}$ and $p_{i,2}$ are the {\em padding rows} of $r_i$.  We also add 2 {\em padding columns} that contain 1-cells in the 4 rows of $r_i$ and $2n$$-$$2$ columns that contain only 0-cells for the 4 rows for of $r_i$. Let $S_i$ be the set of columns with 1-cells in row $r_i$ of $M$.
Let $e_S$ be the row vector of length $m+2n$ with value 1 on all indices
in $S\subseteq [m+2n]$, and 0 everywhere else. The new rows are defined as:
(i) $p_{i,1}$ = $e_{\{m+2i-1\}}$;
(ii) $r_{i,1}$  $= e_{S_i\cup\{m+2i-1\}}$;
(iii) $r_{i,2}$  $= e_{S_i\cup\{m+2i\}}$; and
(iv) $p_{i,2}$ $= e_{\{m+2i\}}$.
We insert these rows in the (i)-(iv) order, for $r_1, ..., r_n$, and refer to this order as the {\em default row ordering} of $M'$. We refer to the column ordering of $M'$ after this transformation as the {\em default column ordering} of $M'$.  
An example transformation from $M$ to $M'$ is shown in Figure \ref{fig:nph-transform}.

\input{figs/ex-nph}

To prove this theorem, it suffices to prove that there exists an ordering $\sigma_c$ on
the columns of $M$ such that $\cb(\sigma_c(M))\leq k$ if and only if there exist orderings
$\sigma'=(\sigma_r'$, $\sigma_c')$ on the rows and columns of $M'$ such that $\sigma'(M')$ admits a
box cover of size $\le k+2n$.
Proving one direction of this claim is simple.
If there exists an ordering $\sigma_c$ on the columns of $M$ such that $\cb(\sigma_c(M)) \le k$,
then set $\sigma_r'$ equal to the default row ordering of $M'$.
Also, set the last $2n$ columns in $\sigma_c'$ equal to the default column ordering of the last $2n$
columns of $M'$. Then, set the first $m$ columns in $\sigma_c'$ equal to $\sigma_c$.
Then, the 1-cells in the first $m$ columns of $\sigma'(M')$ can be covered by at most $k$
boxes, and the 1-cells in the last $2n$ columns can be covered by $2n$ boxes,
for a total box cover size of at most $k+2n$.

Proving the converse is significantly more involved.
Let $\sigma'=(\sigma_r'$, $\sigma_c')$ be an ordering on the rows and columns
of $M'$ such that $\sigma'(M')$ admits a box cover $B$ of size $\le k+2n$.
We start with two definitions.
Two rows $r_{i,j}$ and $r_{k,\ell}$ in $M'$ ($i,k\in[n]$ and $j,\ell\in\{1,2\}$) are \emph{equivalent}
if $r_i$ and $r_k$ are equal rows in $M$ (ie. $r_i$ and $r_k$ have 1-cells in the same columns in $M$).
A \emph{run} of equivalent rows is a sequence $E$ of one or more $r_{i,j}$ rows which are consecutive
in $\sigma_r'$ such that all rows in $E$ are equivalent to one another.
We show a sequence of 6 steps that transform $\sigma'$ to match the default row ordering and except in the first $m$ columns also the default column ordering.
For each step, we prove that we can reorder $\sigma'$ without increasing the number of boxes such that
a claim is true of the reordered $\sigma'(M')$,
assuming that all of the previous claims hold. The proofs of these claims are provided in Appendices~\ref{app:step1}-\ref{app:step6}. 
\begin{enumerate}
\item Every $r_{i,j}$ row can be made adjacent to some equivalent $r_{k,\ell}$ row. (App.~\ref{app:step1})
\item Every run of equivalent $r_{i,j}$ rows can be made to have even length. (App.~\ref{app:step2})
\item Every run of equivalent $r_{i,j}$ rows can be made to have length 2. (App.~\ref{app:step3})
\item The padding rows $p_{i,j}$ can be made adjacent to their matching $r_{i,j}$ rows. (App.~\ref{app:step4})
\item The row order $\sigma_r'$ can be made to exactly match the default row order of $M'$. (App.~\ref{app:step5})
\item The column order $\sigma_c'$ can be made to exactly match the default column order of $M'$ on the last $2n$ columns. (App.~\ref{app:step6})
\end{enumerate}
\kaleb{Two 1-cells $c_1$, $c_2$ are \textit{independent} in $\sigma'(M')$ if there is no box containing $c_1$, $c_2$ that contains only 1-cells. An \textit{independent set} is a set of pairwise independent 1-cells. An independent set in $\sigma'(M')$ of size $S$ implies that the minimum box cover size of $\sigma'(M')$ is at least $S$. We will proceed by constructing a sufficiently large independent set in $\sigma'(M')$.}
After the above 6 steps, $M'$ and $\sigma'(M')$ differ only by the ordering of the first $m$ columns.
In $\sigma'(M')$, the last $2n$ columns contain an independent set of size $2n$, by taking the single
1-cell from each of the $p_{i,j}$ rows. These $2n$ 1-cells are independent from all 1-cells in the
first $m$ columns of $\sigma_c'$.
Let $\sigma_c$ be the ordering of the first $m$ columns in $\sigma_c'$.
We claim the first $m$ columns contain an independent set of size $\cb(\sigma_c(M))$.
First, any two 1-cells in separate 4-row units are independent from one another, because the padding rows
between them contain only 0-cells on the first $m$ columns. If a row of $\sigma_c(M)$ has only one consecutive block, add a 1-cell from the corresponding 4-row unit to the independent set. If a row of $\sigma_c(M)$ has
two consecutive blocks, there are two 1-cells in the first $m$ columns of the corresponding 4-row unit
which are independent from one another. Add both to the independent set. Combining the independent sets
from the first $m$ columns and the last $2n$ columns, we get an independent set of size $\cb(\sigma_c(M))+2n$. Therefore any box cover of $\sigma'(M')$ has at least $\cb(\sigma_c(M))+2n$ boxes.
We assumed $\sigma'(M')$ has a box cover of size $\le k+2n$, which implies
$\cb(\sigma_c(M))+2n \leq k + 2n$, so $\cb(\sigma_c(M)) \leq k$, completing the reduction.
\end{proof}

%% file: figs/ex-nph.tex
\begin{figure}
\vspace{-20pt}
\hspace{30pt}
\scalebox{0.8}{\vbox{
\begin{center}
$M=
\begin{array}{c}
r_1 \\ r_2 
\end{array}
\begin{bmatrix}
1 & 0 & 1 & 0 \\
0 & 0 & 1 & 1 
\end{bmatrix}$
\hspace{20pt}
$M'=
\begin{array}{c}
p_{1,1} \\ r_{1,1} \\ r_{1,2} \\ p_{1,2} \\
p_{2,1} \\ r_{2,1} \\ r_{2,2} \\ p_{2,2} 
\end{array}
\left[\begin{array}{cccc|cc|cc}
0 & 0 & 0 & 0 & \tc 1 & 0 & 0 & 0 \\
\tc 1 & 0 & \tc 1 & 0 & \tc 1 & 0 & 0 & 0 \\
\tc 1 & 0 & \tc 1 & 0 & 0 & \tc 1 & 0 & 0 \\
0 & 0 & 0 & 0 & 0 & \tc 1 & 0 & 0 \\
\hline
0 & 0 & 0 & 0 & 0 & 0 & \tc 1 & 0 \\
0 & 0 & \tc 1 & \tc 1 & 0 & 0 & \tc 1 & 0 \\
0 & 0 & \tc 1 & \tc 1 & 0 & 0 & 0 & \tc 1 \\
0 & 0 & 0 & 0 & 0 & 0 & 0 & \tc 1 \\
\end{array}\right]$
\end{center}
}}
\vspace{-10pt}
\caption{An example of the 2CBMP input matrix $M$ and its corresponding $M'$ matrix.}
\label{fig:nph-transform}
\vspace{-10pt}
\end{figure}

%% file: approx-dom-or.tex
In this section,
we provide a $\tilde{O}(N)$-time approximation algorithm for $\DomOrBox$.
Section~\ref{sec:hyperplanes}
develops some machinery necessary to prove our approximation ratio, and Section~\ref{sec:domain-ordering-approx-alg}
presents our approximation algorithm, ADORA. Section~\ref{sec:tetris-reordered} combines
ADORA, GAMB, and Tetris to state new beyond worst-case bounds for join processing.

\subsubsection{Dividing Relations into Hyperplanes}
\label{sec:hyperplanes}

In the simplest case, suppose the best domain ordering $\sigma^*$ for $Q$ 
yields a minimum box cover of size 1.
Then there is a single gap box $b$ in some relation $\sigma^*(R)$
such that $\mathcal{B}=\{b\}$ forms a box cover for $\sigma^*(Q)$.
Fix an arbitrary attribute $A\in\attr(R)$.
We can partition $\dom(A)$ into two sets:
values which are in the $A$-range spanned by $b$, and values which are not.
Consider the domain ordering $\sigma_A$ obtained by placing all
the domain values spanned by $b$ first (in any order), followed by all other values. Doing this for each
$A\in\mathcal{A}$ yields a domain ordering $\sigma$ which recovers the box $b$ to attain a box cover of size 1.
Intuitively, any domain values for $A$ which lie in the span of the same set of boxes in the minimum box cover should be placed
next to one another.
This can be generalized to an approximation algorithm
for any minimum box cover size.
We begin with definitions needed to formalize this approach.

\begin{definition}[$A$-hyperplane]
Let $R\in\mathcal{R}$ be over a set of attributes $\attr(R)$. Let $A\in\attr(R)$ and $a\in\dom(A)$. The $A$-hyperplane of $R$ defined by $a$ is the relation $H(R,A,a)=\pi_{\attr(R)\setminus\{A\}}(\sigma_{A=a}(R))$ 
\kaleb{if $|\attr(R)|>1$ and $H(R,A,a)=\{|\sigma_{A=a}(R)|\}$ if $|\attr(R)|=1$.}
\end{definition}

Let $n_R=|\attr(R)|$. The $A$-hyperplane defined by $a$ in $R$ can be thought of as the ``slice'' of the $n_R$-dimensional
space occupied by $R$ containing only the $(n_R$$-$$1)$-dimensional subspace where the attribute $A$ is
fixed to the value $a$. This is a natural generalization of ``rows'' and ``columns'' which were useful for
discussing 2-dimensional relations in Section~\ref{sec:nph}.

\begin{definition}[Equivalent domain values]
Let $Q=(\mathcal{R},\mathcal{A})$ be a query, let $A\in\mathcal{A}$, and let $a_1,a_2\in\dom(A)$.
$a_1$ and $a_2$ are \emph{equivalent in $Q$} if for all $R\in\mathcal{R}$ we have $H(R,A,a_1)=H(R,A,a_2)$.
In this case, we write $a_1\sim a_2$.
\end{definition}

For $a\in\dom(A)$, the subset of domain values $\Eq(a)=\{a'\in\dom(A):a\sim a'\}\subseteq\dom(A)$ is
called the \emph{equivalence class} of $a$. The equivalence classes for all of the values in $\dom(A)$
form a partition of $\dom(A)$. 
The next lemma bounds the number of these equivalence classes as a function of the minimum box cover size of any domain ordering $\sigma$.

\begin{lemma}
\label{lemma:few-boxes-few-planes}
Let $\sigma$ be a domain ordering for $Q=(\mathcal{R},\mathcal{A})$. 
Let $A$ be an attribute in $\mathcal{A}$ and $h$ be the number of equivalence classes of the values in $\dom(A)$.  Then $h \le 2\cdot K_{\Box}(\sigma(Q))+1$.
\end{lemma}

\input{figs/ex-hyperplanes}

\begin{proof}
Let $A\in\mathcal{A}$ and let $a_1,a_2\in\dom(A)$ be such that $a_1$ directly precedes $a_2$ in $\sigma_A$ 
and $a_1\not\sim a_2$. We refer to the $a_1$, $a_2$ boundary as a ``switch'' along $A$. 
Observe that there are at least $h$$-$$1$ switches along $A$. This minimum is attained when the values in each equivalence class are placed in a single consecutive run in $\sigma_A$. 
Since $a_1\not\sim a_2$, there is some relation $R\in\mathcal{R}$ such that
$H_1=H(R,A,a_1)\neq H(R,A,a_2)=H_2$. 
Then there is some tuple $t$ which is in $H_1$ but not $H_2$ or vice versa. Assume w.l.o.g. that $t\in H_1$ and $t\not\in H_2$. 
Let $t_1 = \langle a_1, t\rangle$ and $t_2 = \langle a_2, t\rangle$ be the tuples that extend $t$ to attribute $A$ with values $a_1$ and $a_2$, respectively. 
This means that $t_1 \in R$ and
$t_2\not\in R$. 
Let $\mathcal{B}$ be a box cover for $\sigma(Q)$ with $K_{\Box}(\sigma(Q))$ boxes. Let $\mathcal{B}_{R}$ be the set of boxes in $\mathcal{B}$ that are from $R$ and cover the complement of $R$ (so $|\mathcal{B}_{R}| \le K_{\Box}(\sigma(Q))$). 
Let $b\in \mathcal{B}_R$ be a box
covering $t_2$ (and not $t_1$ since $b$ is a gap box).
\kaleb{In this context, a \textit{face} of $b$ along the $A$ axis is one of the two distinct portions of the boundary of $b$ contained in a hyperplane orthogonal to the $A$ axis that does not contain any interior points of $b$.}
Since $a_1$ and $a_2$ are adjacent in $\sigma_A$, one face of $b$ along the $A$ axis is the $(a_1, a_2)$ switch, i.e. one face of $b$ lies on the boundary between $H_1$ and $H_2$. Every box $b$ has exactly two faces along $A$, so there are $\leq 2K_{\Box}(\sigma(Q))$ faces of boxes in $\mathcal{B}$ along the $A$ axis. Two different switches cannot correspond to the same face of the same box.
As an example, Figure~\ref{fig:hyperplanes-1} shows the switches in attribute $A$ and the faces of gap boxes that these switches correspond to, which are highlighted in colour.
This completes the argument that each ($a_1$, $a_2$) switch corresponds to a distinct face of some box along the $A$ axis. There are at least $h$$-$$1$ switches and at most $2K_{\Box}(\sigma(Q))$ box faces, so $h \le 2K_{\Box}(\sigma(Q))$$+$$1$.
\end{proof}

\subsubsection{ADORA}
\label{sec:domain-ordering-approx-alg}
Lemma~\ref{lemma:few-boxes-few-planes} inspires an approximation algorithm for $\DomOrBox$.
Let $\sigma^*$ be the optimal domain ordering for
$\DomOrBox$ on $Q$.
Let $K=K_{\Box}(\sigma^*(Q))$ throughout this section. 
Algorithm~\ref{alg:adora} presents
the pseudocode for our \textbf{A}pproximate \textbf{D}omain \textbf{Or}dering \textbf{A}lgorithm (ADORA).
ADORA uses Algorithm~\ref{alg:order-attr} as a subroutine to produce an ordering $\sigma_A$ for $\dom(A)$ that contains each equivalence class in $\dom(A)$ as a consecutive run. 

\begin{theorem}
\label{thm:domain-ordering-approx}
Let $Q=(\mathcal{R},\mathcal{A})$ be a query and $\sigma^*$ be an optimal domain ordering for $\DomOrBox$ on $Q$.
Let $K=K_{\Box}(\sigma^*(Q))$.
Then ADORA produces a domain ordering $\sigma$ in $\tilde{O}(N)$ time
such that $K_{\Box}(\sigma(Q))=\tilde{O}(K^r)$,
where $r$ is the maximum arity of a relation in $\mathcal{R}$.
\end{theorem}

\begin{proof}
We defer the runtime analysis of ADORA to Appendix~\ref{app:adora-runtime} and prove the approximation ratio here.
We begin by arguing that given an attribute $A$, the ordering returned by Algorithm~\ref{alg:order-attr}  
places every equivalence class of $\dom(A)$ in one consecutive run. 
The for-loop on line~\ref{line:subroutinesecondfor} iterates over
each $a\in\dom(A)$ that appears in $Q$ and constructs an array $\mathcal{T}[a]$.  
$\mathcal{T}[a]$ is the result of appending the $A$-hyperplanes
$H(R,A,a)$ for each $R\in\mathcal{S}$ in a fixed order.
Line~\ref{line:subroutine:firstfor} sorts each relation lexicographically starting with $A$ (notice that the order $\phi$ is defined to place $A$ first).
These two facts ensure that after the for-loop on line~\ref{line:subroutinesecondfor} has finished, $\mathcal{T}[a_1]=\mathcal{T}[a_2]$
if and only if $a_1\sim a_2$.
The final sort of $D$ on line~\ref{line:subroutinefinalsort} sorts values of $\dom(A)$, say $a_i$ and $a_j$, according to the lexicographic order of $T[a_i]$ and $T[a_j]$, \kaleb{so the lists are compared item by item from start to end, with each item compared using attribute ordering $\phi$.} This ensures all $A$ values in the same equivalence class will be in one consecutive run in $\sigma_A$.

The output of ADORA is a domain ordering $\sigma$, which orders each attribute $A \in \mathcal{A}$ according to the $\sigma_A$ returned by Algorithm~\ref{alg:order-attr}. We next prove there exists a box cover for $\sigma(Q)$ of size $\tilde{O}(K^r)$.
Let $R\in\mathcal{R}$. Suppose $|\attr(R)|=n_R$ and note that $n_R\leq r$.
Let $A\in\attr(R)$.
Lemma~\ref{lemma:few-boxes-few-planes} states that $\dom(A)$ contains at most $2K+1$ equivalence classes, which we proved are placed consecutively in $\sigma_A$.
By definition, if $a_1\sim a_2$, then $H(R,A,a_1)=H(R,A,a_2)$. Therefore $\sigma_A$ consists of
a sequence of at most $2K+1$ consecutive runs of $A$-values where the values in each run have identical $A$-hyperplanes
in $R$. This holds for all $A\in\attr(R)$. The runs of identical hyperplanes partition the $n_R$-dimensional space of $\sigma(R)$ into
at most $(2K+1)^{n_R}$ many $n_R$-dimensional grid boxes. Each dimension of a grid box is formed by one of the (at most) $2K+1$ runs from one attribute. By construction, these grid boxes form a partition of the  $n_R$-dimensional space as each grid box is a distinct combination of equivalence classes for the attributes and the orderings returned by Algorithm~\ref{alg:order-attr} cover all the values in $\dom(A)$. 
Figure~\ref{fig:hyperplanes-2} demonstrates the grid boxes implied by the equivalence classes in the orderings of a relation. In the figure, there are two equivalence classes for attribute $A$ and three for $B$, dividing the relation into 6 grid boxes. 

We argue that each grid box is completely full of either gaps or tuples.
Let $t\in R$, let $g$ be the grid box containing $t$, and let $t'$ be another point in $g$. \kaleb{Two points $t_1$, $t_2$ are {\em adjacent} if for some attribute $A$, $t_1.A$ and $t_2.A$ are adjacent in $\sigma_A$.} Consider moving from $t$ to $t'$ through any sequence of adjacent points in $g$.
When we pass through a point we are moving from one $A$-hyperplane to an identical $A$-hyperplane for some attribute $A$. Thus every point along this path must also be a tuple in $R$. A similar argument for gaps implies that every point in a grid box that contains one gap must also be a gap. Since the grid boxes partition the domain of $R$, constructing one box for each gap grid box results in a box cover $\mathcal{B}_R$ for $R$.
Since there are at most $(2K+1)^{n_R}$ grid boxes, $|\mathcal{B}_R|\leq(2K+1)^{n_R}$.
We can construct such a box cover for each $R\in\mathcal{R}$ to obtain a box cover for $\sigma(Q)$ of size
$\sum_{R\in\mathcal{R}}(2K+1)^{n_R}\leq m(2K+1)^r=\tilde{O}(K^r)$, completing the proof of ADORA's approximation ratio. 
\end{proof}

\begin{algorithm}[t]
\caption{\textsc{ADORA}($Q=(\mathcal{R},\mathcal{A})$): Computes a domain ordering.}
\label{alg:adora}
\begin{algorithmic}[1]
\For{$A\in\mathcal{A}$}
  \State $\sigma_A :=$ \textsc{OrderAttr}($Q,A$)
\EndFor
\State\Return $\sigma=\{\sigma_A\}_{A\in\mathcal{A}}$
\end{algorithmic}
\end{algorithm}

\begin{algorithm}[t]
\begin{algorithmic}[1]
  \State $\phi :=$ any attribute ordering of $\mathcal{A}$ which places $A$ first
  \State $\mathcal{S} := \{R\in\mathcal{R}:A\in\attr(R)\}$, $D := \bigcup_{R\in\mathcal{S}}\pi_A(R)$, $\mathcal{T}:=\emptyset$ \label{line:D}
  \For{$R\in\mathcal{S}$}
    \State Sort $R$ lexicographically according to $\phi$ \label{line:subroutine:firstfor}
  \EndFor
  \For{$a\in D$} \label{line:subroutinesecondfor}
    \State $\mathcal{T}[a] := []$
    \For{$R\in\mathcal{S}$ in a fixed order}
      \State $\mathcal{T}[a]$.append($H(R,A,a)$)
    \EndFor
  \EndFor
  \State Sort $D$ by ordering $a_i$ and $a_j$ according to the lexicographic order of $\mathcal{T}[a_i]$ and  $\mathcal{T}[a_j]$ \label{line:subroutinefinalsort}
  \State\Return $\sigma_A=D$ (append $a \not \in D$ to $\sigma_A$ in arbitrary order)
\end{algorithmic}
\caption{\textsc{OrderAttr}($Q, A$): Groups equivalence classes for $A$ into consecutive runs.}
\label{alg:order-attr}
\end{algorithm}

Appendix~\ref{sec:adora-tight} shows that our analysis of ADORA's approximation factor is asymptotically tight by defining a family of queries over binary relations which have orderings with box covers of size $K$, whereas the orderings that ADORA returns require $\Omega(K^2)$ boxes.

\subsection{TetrisReordered}
\label{sec:tetris-reordered}

We next combine ADORA, GAMB, and Tetris in a new join algorithm we call {\em TetrisReordered} to obtain new beyond worst-case optimal results for join queries. Algorithm~\ref{alg:tetris-reordered} presents the pseudocode of TetrisReordered.
Corollary~\ref{cor:tetris-reordered-width} immediately follows from Theorems~\ref{thm:gamb} and~\ref{thm:domain-ordering-approx} from this paper, and Theorems~4.9 and~4.11 from reference~\cite{tetris}. 

\begin{corollary}
\label{cor:tetris-reordered-width}
Let $Q=(\mathcal{R},\mathcal{A})$ be a join query. Let $w$ be the treewidth of $Q$, $n=|\mathcal{A}|$, and $r$ the maximum
arity of a relation in $\mathcal{R}$. Let $\sigma^*$ be an optimal
solution to $\DomOrBox$ on $Q$ and let $K=K_{\Box}(\sigma^*(Q))$. TetrisReordered
computes $Q$ 
in $\tilde{O}(N+K^{r(w+1)}+Z)$ time by using Tetris-Reloaded
or in $\tilde{O}(N+K^{rn/2}+Z)$ time by using Tetris-LoadBalanced as a subroutine.
\end{corollary}

\begin{algorithm}
\caption{ \textsc{TetrisReordered}($Q$):}
\label{alg:tetris-reordered}
\begin{algorithmic}[1]
\State $\sigma :=$ \textsc{ADORA}($Q$) (Algorithm \ref{alg:adora})
\For{$R\in\mathcal{R}$}
  \State $\mathcal{B}_R :=$ \textsc{GAMB}($\sigma(Q)$) (Algorithm \ref{alg:gamb})
\EndFor
\State\Return $\sigma^{-1}$(\textsc{Tetris}($\mathcal{B}=\{\mathcal{B}_R\}_{R\in\mathcal{R}}$))
\end{algorithmic}
\end{algorithm}



We next show that there are infinite families of queries for which these bounds are arbitrarily smaller than prior bounds stated for Tetris in reference~\cite{tetris}. In fact,
given an arbitrary query $Q$ with any number of output tuples and any certificate size, and a default domain ordering $\sigma$, we can generate families of queries for which TetrisReordered is unboundedly faster than Tetris on $\sigma$. Our method can be seen as a generalization of the ``checkerboard'' example in Figure~\ref{fig:reorder-rows-cols}.\footnote{The example in Figure~\ref{fig:reorder-rows-cols} is a simplified version of the one used to prove Lemma~J.1 in reference~\cite{tetris}, which shows that {\em general resolutions}, which are logical operations on two DNF clauses, are more powerful than geometric resolutions, which are constrained to contiguous geometric boxes.}
Take an arbitrary query $Q=(\mathcal{R},\mathcal{A})$ with $N$ input tuples and $Z$ output tuples. Recall that the minimum certificate size for $Q$ under its default ordering is denoted $C_{\Box}(Q)$. Let $\sigma_{ADR}$ be the  ordering ADORA generates on $Q$ and let $U$ be the
corresponding upper bound on $K_{\Box}(\sigma_{ADR}(Q))$ provided by Theorem~\ref{thm:domain-ordering-approx}. Recall that $U$ depends on $r$, the maximum arity of a relation in $\mathcal{R}$,
and is an upper bound on the minimum box cover and certificate size for $\sigma_{ADR}(Q)$. We generate a family of queries $Q_p$ from $Q$ for $p=1,2,...$, whose minimum certificate size (under $\sigma$) increases to $2^{rp}C_{\Box}(Q)$ but the upper bound provided by ADORA according to Theorem~\ref{thm:domain-ordering-approx} remains at $U$.

Let $\mathcal{A}_p$ be the attribute set obtained from $\mathcal{A}$ by adding, for each $A\in\mathcal{A}$, an additional $p$ bits as a prefix
to the $d$ bits of $A$. For every relation $R\in\mathcal{R}$, construct a relation $R_p$ with $\attr(R_p)\subseteq\mathcal{A}_p$ corresponding
to $\attr(R)$. For each $t\in R$, add the following tuples to $R_p$:
\begin{center}
$\{\langle p_At.A\rangle_{A\in\attr(R)}:p_A\in\{0,1\}^p$ $\forall A\in\attr(R)\}$
\end{center}

The $p$ bits added to each attribute do not affect the structure of the query,
since these bits vary over all possible valuations for each tuple from the original query. For each attribute $A$, these bits 
effectively create $2^p$ ``copies'' of
each $A$-hyperplane. 
This increases the size of the query's input, output, and box certificate under the default ordering.
The query $Q_p=(\mathcal{R}_p,\mathcal{A}_p)$ has input size $2^{rp}N$,
output size $2^{np}Z$, and minimum box certificate
size $2^{rp}C_{\Box}(Q)$, where $r$ is the maximum arity of a relation in $\mathcal{R}$ and $n=|\mathcal{A}|$. However, this construction does not affect the number of equivalence classes on any dimension. Instead, it increases the size of each equivalence class by a factor of $2^p$. To see this, consider two values of an attribute $A\in\mathcal{A}$, $a_1$ and $a_2$, that were in the same equivalence class in $Q$. 
That is, they had the same $A$-hyperplanes for every relation $R\in\mathcal{R}$ such that $A\in\attr(R)$.
After adding the $p$ bits, there will be $2^p$ ``copies'' of $a_1$ and $a_2$, one for each $p$-bit prefix that was appended to tuples that contained $a_1$ and $a_2$. Each copy will have the same (but larger) $A$-hyperplane.
The number of equivalence classes on each attribute will remain the same, so the bound of
Theorem~\ref{thm:domain-ordering-approx} on $K_{\Box}(\sigma_{ADR}(Q))$ will remain at $U$. 
As $p$ increases, the performance gap between TetrisReordered and prior versions of Tetris becomes arbitrarily large.

%% file: figs/ex-hyperplanes.tex
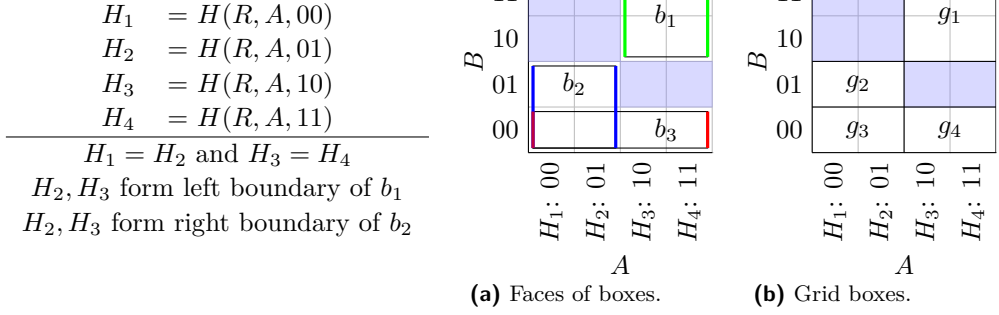
\begin{figure*}

\vspace{-15pt}
\begin{center}
\begin{tabular}{c@{\hskip 10pt}c@{\hskip 10pt}c}
\begin{tabular}{c}
\begin{tabular}{rl}
$H_1$ & = $H(R, A, 00)$ \\
$H_2$ & = $H(R, A, 01)$ \\
$H_3$ & = $H(R, A, 10)$ \\
$H_4$ & = $H(R, A, 11)$
\end{tabular} \\ \hline
$H_1=H_2$ and $H_3=H_4$ \\
$H_2,H_3$ form left boundary of $b_1$ \\
$H_2,H_3$ form right boundary of $b_2$
\end{tabular}
&
\begin{tikzpicture}[baseline=(current bounding box.center)]
\begin{axis}[axis lines=left,xlabel=$A$,ylabel=$B$,
  xmin=0,xmax=4,
  ymin=0,ymax=4,
  tickwidth=0,
  xticklabel interval boundaries,yticklabel interval boundaries,
  xtick={0,1,2,3,4},ytick={0,1,2,3,4},
  xticklabels={$H_1$: 00, $H_2$: 01,$H_3$: 10,$H_4$: 11},yticklabels={00,01,10,11},
  xticklabel style={rotate=90},
  y label style={at={(axis description cs:0.11,0.5)},anchor=north},
  x label style={at={(axis description cs:0.5,-0.5)},anchor=south},
  grid=both,axis line style={-},
  height=4cm,width=4cm]
\filldraw[\tuplecolour,opacity=\tupleopacity] (axis cs:0,2) rectangle (axis cs:2,4);
\filldraw[\tuplecolour,opacity=\tupleopacity] (axis cs:2,1) rectangle (axis cs:4,2);
\draw[\boxcolour] (axis cs:2.1,2.1) rectangle (axis cs:3.9,3.9);
\draw[\boxcolour] (axis cs:0.1,0.1) rectangle (axis cs:1.9,1.9);
\draw[\boxcolour] (axis cs:0.1,0.1) rectangle (axis cs:3.9,0.9);
\node at (axis cs:3,3) {$b_1$};
\node at (axis cs:1,1.5) {$b_2$};
\node at (axis cs:3,0.5) {$b_3$};
\draw[green,very thick] (axis cs:2.1,2.1) -- (axis cs:2.1,3.9);
\draw[green,very thick] (axis cs:3.9,2.1) -- (axis cs:3.9,3.9);
\draw[blue,very thick] (axis cs:0.1,0.1) -- (axis cs:0.1,1.9);
\draw[blue,very thick] (axis cs:1.9,0.1) -- (axis cs:1.9,1.9);
\draw[purple,very thick] (axis cs:0.1,0.1) -- (axis cs:0.1,0.9);
\draw[red,very thick] (axis cs:3.9,0.1) -- (axis cs:3.9,0.9);
\end{axis}
\end{tikzpicture}
&
\begin{tikzpicture}[baseline=(current bounding box.center)]
\begin{axis}[axis lines=left,xlabel=$A$,ylabel=$B$,
  xmin=0,xmax=4,
  ymin=0,ymax=4,
  tickwidth=0,
  xticklabel interval boundaries,yticklabel interval boundaries,
  xtick={0,1,2,3,4},ytick={0,1,2,3,4},
  xticklabels={$H_1$: 00, $H_2$: 01,$H_3$: 10,$H_4$: 11},yticklabels={00,01,10,11},
  xticklabel style={rotate=90},
  y label style={at={(axis description cs:0.11,0.5)},anchor=north},
  x label style={at={(axis description cs:0.5,-0.5)},anchor=south},
  grid=both,axis line style={-},
  height=4cm,width=4cm]
\filldraw[\tuplecolour,opacity=\tupleopacity] (axis cs:0,2) rectangle (axis cs:2,4);
\filldraw[\tuplecolour,opacity=\tupleopacity] (axis cs:2,1) rectangle (axis cs:4,2);
\draw[black] (axis cs:2,0) -- (axis cs:2,4);
\draw[black] (axis cs:4,0) -- (axis cs:4,4);
\draw[black] (axis cs:0,1) -- (axis cs:4,1);
\draw[black] (axis cs:0,2) -- (axis cs:4,2);
\draw[black] (axis cs:0,4) -- (axis cs:4,4);
\node at (axis cs:3,3) {$g_1$};
\node at (axis cs:1,1.5) {$g_2$};
\node at (axis cs:1,0.5) {$g_3$};
\node at (axis cs:3,0.5) {$g_4$};
\end{axis}
\end{tikzpicture} \\
&
\begin{subfigure}{85pt}
\caption{Faces of boxes.}
\label{fig:hyperplanes-1}
\end{subfigure}
&
\begin{subfigure}{85pt}
\caption{Grid boxes.}
\label{fig:hyperplanes-2}
\end{subfigure}
\end{tabular}
\end{center}
\vspace{-20pt}
\caption{An illustration of how $A$-hyperplane switches form the boundaries of the gap boxes of $R$ (left)
and how dividing $R$ into grid cells defined by hyperplane switches induces a box cover (right).}
\label{fig:hyperplanes}
\vspace{-5pt}
\end{figure*}

%% file: future-work.tex
\vspace{-5pt}
\section{Open Problems}
\label{chap-future-work}

\kaleb{The problems defined in this paper are ripe for further study. Beyond NP-hardness, we know little about $\DomOrCert$. Even if the domain ordering is fixed, we do not know of a way to approximate the minimum certificate size that is asymptotically faster than computing the join.} Appendix~B.3 in reference~\cite{tetris} describes how to use a variant of Minesweeper to compute a certificate. \kaleb{In Appendix~G of the online version of this paper~\cite{this-arxiv}, we show a variant of Tetris can do the same.}
\kaleb{These approaches effectively compute the join to compute a certificate.}
A difficult aspect of this problem is that a certificate is a box cover for the output relation using boxes from the input relations. However, since the output tuples are not known a priori, the exact space that needs to be covered is not known at domain reordering time. 
It is also not known a priori which input tuples are part of the output and therefore which gap boxes from the input relations are part of the certificate.

\kaleb{Similarly, little is known about the following problems: (1) determining whether a specific gap box is in the minimum-size certificate under a fixed domain ordering; (2) verifying that a given domain ordering induces a box cover or certificate of minimum size; and (3) variants of $\DomOrBox$ and $\DomOrCert$ under additional assumptions about the structure of the input relations.}
As an example, Appendix~\ref{sec:semi-join-certificate} shows that if the input relations are fully semi-join reduced (a problem known to be hard for cyclic queries~\cite{bernstein:semi-joins}), so all input tuples are part of the output and all gap boxes are relevant to the certificate, then the minimum box cover and certificate sizes are within an $\tilde{O}(1)$ factor of one another. In this special case, ADORA approximates $\DomOrCert$.

\kaleb{Each of these open problems can be defined over general or dyadic boxes, creating two related but distinct problems. Sufficiently strong hardness of approximation results for either version of a problem would imply the difference is negligible, since the solutions would be within an $\tilde{O}(1)$-factor of one another by Lemma~\ref{lemma:dyadic}. Theorem~\ref{thm:nph} shows that the general box versions of $\DomOrBox$ and $\DomOrCert$ are NP-hard, but this theorem does not imply a hardness result about the dyadic versions of these problems. Proving a hardness of approximation result is a direction for future study. Until such a result is known, the distinction between general and dyadic boxes is crucial to studying these problems.}

\kaleb{Developing an ADORA-like preprocessing algorithm that provides similar results but is query-independent, or has a better approximation ratio, is also a direction for future research. ADORA runs in $\tilde{O}(N)$ time, so its query-dependence precludes the possibility of a sub-linear time join algorithm using ADORA. The domain ordering must be computed from scratch for each different query, even if the relations in the database have not changed. Sharing some of this computation between different queries would improve on our results.}

%% file: conclusions.tex
\vspace{-5pt}
\section{Conclusions}
\label{chap-conclusions}

For queries with fixed domain orderings,
we established a $\tilde{O}(N)$-time algorithm GAMB to create a single globally good box cover index which is guaranteed to contain a certificate at most a $\tilde{O}(1)$ factor larger than the minimum size certificate for any box cover.
We then studied $\DomOrBox$, the problem of finding a domain ordering that yields the smallest possible box cover size for a given query $Q$. We proved that $\DomOrBox$ is NP-hard and presented a $\tilde{O}(N)$-time approximation algorithm ADORA that computes an ordering which yields a box cover of size $\tilde{O}(K^r)$, where $K$ is the minimum box cover size under any ordering and $r$ is the maximum arity of any relation in $Q$. 
We combined ADORA, GAMB, and Tetris in an algorithm we call TetrisReordered and
stated new beyond worst-case optimal runtimes for join processing in Corollary~\ref{cor:tetris-reordered-width}. TetrisReordered can improve the known performance bounds of prior versions of Tetris (on any fixed ordering) on infinite families of queries. 

Our work leaves several interesting problems open
\kaleb{as discussed in Section~\ref{chap-future-work}. Our results are limited to the problems $\DomOrBox$ and $\DomOrCert$, and there are several interesting variants of these problems for which little is known.}

%% file: appendix.tex
\vspace{-10pt}
\section{Example with \texorpdfstring{$\omega(N)$}{N\^2} Maximal General Gap Boxes}
\label{sec:many-general-boxes}

In Section~\ref{sec:generating-boxes}, GAMB generates all maximal dyadic gap boxes
in $\tilde{O}(N)$ time. The use of dyadic boxes instead of general boxes in GAMB is necessary,
because there are relations for which the number of maximal general gap boxes is asymptotically greater than
the number of tuples in the relation. 
Our construction generalizes the example in Figure 15 in Appendix B.3 of reference~\cite{tetris}.
Let $N$ be an even number, $A$ and $B$ be attributes over domains of size $N$, and
$R_N$ be the following relation.
\begin{center}
$R_N(A,B)=\{\langle i,N/2$$-$$i$$-$$1\rangle :0\leq i<N/2\}\cup\{\langle N/2$$+$$i,N$$-$$i$$-$$1\rangle:0\leq i<N/2\}$
\end{center}

Consider the following sets of tuples which are \emph{not} in $R_N$.
\begin{center}
$T_N=\{t_i=\langle i,N/2$$-$$i\rangle:0\leq i\leq N/2\}$

$S_N=\{s_i=\langle N/2$$+$$i$$-$$1,N$$-$$i$$-$$1\rangle:0\leq i\leq N/2\}$
\end{center}

\input{figs/ex-many-general-boxes}

Figure \ref{fig:many-general-boxes} depicts $R_8$, $T_8$, and $S_8$.
In this diagram, a set of 5 maximal general gap boxes with bottom left corners at
$t_2$ is depicted. The top right corners of these boxes are the 5 tuples in $S_8$.
In fact, for each $0\leq i\leq 4$, there are 4 or 5 maximal
general gap boxes in $R_8$ with their bottom left corner at $t_i$.
This property generalizes from $R_8$ to any value of $N$.
For each $0\leq i\leq N/2$, there are at least $N/2$ maximal general gap boxes in $R_N$ with
their bottom left corner at $t_i$.
Since there are $N/2+1$ tuples in $T_N$, the total number of maximal general gap boxes
in $R_N$ is at least $(N/2+1)(N/2)=\Theta(N^2)=\omega(N)$.

\vspace{-10pt}
\section{Incremental Maintenance of a Maximal Box Cover Index}
\label{sec:gamb-index-maintenance}

\kaleb{
Theorem~\ref{thm:gamb} states that running GAMB on a relation $R$ produces a set of dyadic gap boxes containing all maximal dyadic gap boxes of $R$. In this section, we refer to such a set as a \textit{maximal dyadic box cover index (MDBCI)} for $R$. An MDBCI can be thought of as an index for $R$ that can be computed once and used in any query over $R$. Traditional database indexes are useful because they are easy to incrementally maintain when tuples are added to or removed from a relation, so the index does not need to be computed from scratch every time the relation is modified. The following results show that efficient incremental maintenance is also possible with MDBCIs.
\iftoggle{appendix-online}{}{The proofs of Theorems~\ref{thm:mdbci-insert},~\ref{thm:mdbci-delete}, and~\ref{thm:mdbci-runtime} can be found in Appendix~B of the online version of this paper~\cite{this-arxiv}.}

\begin{algorithm}[t]
\caption{InsertMDBCI($R, B, t$): Update $B$ after $t$ is inserted to $R$}
\label{alg:mdbci-insert}
\begin{algorithmic}[1]
\For{each $b\in B$ such that $t\in b$} \label{mdbci-insert:first-loop}
  \State $B:= B\setminus\{b\}$ \label{mdbci-insert:remove}
  \For{each $b'$ such that $t\in b'\subseteq b$} \label{mdbci-insert:second-loop}
    \For{$A\in\attr(R)$ such that $|b.A|<d$} \label{mdbci-insert:third-loop}
      \State Let $b''$ be the box when one bit is added to $b'.A$ such that $t\not\in b''$ \label{mdbci-insert:inner}
      \State $B:= B\cup\{b''\}$ \label{mdbci-insert:insert}
    \EndFor
  \EndFor
\EndFor
\end{algorithmic}
\end{algorithm}

\begin{theorem}
\label{thm:mdbci-insert}
Suppose that $R$ is a relation with MDBCI $B$. Let $t\not\in R$. Let $R'=R\cup\{t\}$ and let $B'$ be the result of running Algorithm~\ref{alg:mdbci-insert} with $R'$, $B$ and $t$ as input. Then $B'$ is an MDBCI for $R'$.
\end{theorem}

\iftoggle{appendix-online}{
\begin{proof}
First, we will show that all maximal dyadic gap boxes of $R'$ are in $B'$. Let $b''$ be a maximal dyadic gap box of $R'$. If $b''$ is also a maximal dyadic gap box of $R$, then $b''\in B'$ because $b''\in B$ and $b''$ was not removed on line~\ref{mdbci-insert:remove}. If $b''$ is not a maximal dyadix gap box of $R$, then there must be some maximal dyadic gap box $b\in B$ for $R$ such that $b''\subset b$. Since $R$ and $R'$ differ only by $t$, $t\in b$ and $t\not\in b''$. Since $b''$ is maximal in $R'$ and not in $R$, there is an attribute $A\in\attr(R)$ for which removing the last bit of $b''.A$ results in a box $b'$ that satisfies $b''\subset b' \subseteq b$ and $t\in b'$. Since $b\in B$, InsertMDBCI iterates over $b$ on line~\ref{mdbci-insert:first-loop}. Since $t\in b'$ and $b'\subseteq b$, InsertMDBCI iterates over $b'$ on line~\ref{mdbci-insert:second-loop}. Since $|b'.A|<d$, $t\not\in b''$, and $b''$ can be obtained from $b'$ by adding one bit to the end of $b'.A$, InsertMDBCI iterates over $A$ on line~\ref{mdbci-insert:third-loop} and constructs $b''$ on line~\ref{mdbci-insert:inner}. Therefore, $b''\in B'$.

It remains to prove that $B'$ does not contain any boxes covering tuples of $R'$. Any boxes from $B$ which covered $t$ were removed on line~\ref{mdbci-insert:remove} and no other boxes in $B$ cover any tuples of $R'$. On line~\ref{mdbci-insert:inner}, all boxes $b''$ added to $B'$ do not cover $t$ by construction. Furthermore, these boxes are contained in some box $b\in B$, so they do not cover any other tuples in $R'$.
\end{proof}
}{}

\begin{algorithm}[t]
\caption{DeleteMDBCI($R, B, t$): Update $B$ after $t$ is deleted from $R$}
\label{alg:mdbci-delete}
\begin{algorithmic}[1]
\For{every dyadic box $b$ such that $t\in b$} \label{mdbci-delete:first-loop}
  \State addb $:=$ True
  \For{each $b'$ such that $t\in b'\subset b$} \label{mdbci-delete:second-loop}
    \For{$A\in\attr(R)$ such that $b'.A\neq *$} \label{mdbci-delete:third-loop}
      \State Let $b''$ be the box when the last bit of $b'.A$ is flipped \label{mdbci-delete:inner}
      \State \algorithmicif\hspace{1mm} there is no box in $B$ that contains $b''$ \algorithmicthen\hspace{1mm} addb $:=$ False \label{mdbci-delete:box-contains}
    \EndFor
  \EndFor
  \If{addb}
    \State $B:=B\cup\{b\}$ \label{mdbci-delete:add}
  \EndIf
\EndFor
\end{algorithmic}
\end{algorithm}

\begin{theorem}
\label{thm:mdbci-delete}
Suppose that $R$ is a relation with maximal dyadic box cover index $B$. Let $t\in R$. Let $R'=R\setminus\{t\}$ and let $B'$ be the result of running Algorithm~\ref{alg:mdbci-delete} with $R'$, $B$ and $t$ as input. Then $B'$ is an MDBCI for $R'$.
\end{theorem}

\iftoggle{appendix-online}{
\begin{proof}
Let $b$ be a maximal dyadic gap box of $R'$. If $b$ is also a maximal dyadic gap box of $R$, then $b\in B'$ because $b\in B$ and $B\subseteq B'$. Let $b'$ be any box such that $t\in b'\subset b$. For any $A$ such that $b'.A\neq *$, let $b''$ be the box when the last bit of $b'.A$ is flipped. Since $b$ is a gap box for $R'$, $R\setminus R'=\{t\}$, and $t\not\in b''$, $b''$ is also a gap box for $R$. Since $B$ is an MDBCI for $R$ and $b''$ is a gap box for $R$, there is a box in $B$ that contains $b''$. Combining these observations, we can see that DeleteMDBCI will iterate over $b$ on line~\ref{mdbci-delete:first-loop}, and for all iterations of lines~\ref{mdbci-delete:second-loop} and~\ref{mdbci-delete:third-loop}, the condition on line~\ref{mdbci-delete:box-contains} will evaluate to false. Therefore, $b\in B'$.

It remains to prove that no box $b\in B'$ contains a tuple of $R'$. If $b$ is also in $B$, then $b$ contains no tuples of $R'$ since $R'\subset R$. If $b\not\in B$, then $b$ was added by DeleteMDBCI on line~\ref{mdbci-delete:add}. Suppose $b$ contains a tuple $t'\in R'$. Let $b''\subset b$ be a dyadic box that satisfies $t'\in b''$, $t\not\in b''$, $b''\subset b$, and is maximal in the sense that there is no attribute $A\in\attr(R)$ for which the last bit of $b''.A$ can be removed while still satisfying these conditions. Since $b''\subset b$, $t\not\in b''$, and $t\in b$, there exists $A\in\attr(R)$ such that removing the last bit of $b''.A$ creates a box that contains $t$. Let $b'\subset b$ be the box when the last bit of $b''.A$ is flipped. Note that $t\in b'$, so DeleteMDBCI iterated over $b'$ on line~\ref{mdbci-delete:second-loop}. Then, DeleteMDBCI iterated over $A$ on line~\ref{mdbci-delete:third-loop} and constructed $b''$ on line~\ref{mdbci-delete:inner}. Since $t'\in b'$, $t\in R$, and $B$ is an MDBCI for $R$, there is no box in $B$ that contains $b''$, so the condition on line~\ref{mdbci-delete:box-contains} ensures $b''\not\in B'$. This is a contradiction, therefore no box $b\in B'$ contains any tuple of $R'$.
\end{proof}
}{}

\begin{theorem}
\label{thm:mdbci-runtime}
Algorithms~\ref{alg:mdbci-insert} and~\ref{alg:mdbci-delete} run in $\tilde{O}(1)$ time.
\end{theorem}

\iftoggle{appendix-online}{
\begin{proof}
First, note that Appendix C of reference~\cite{tetris} describes a data structure to store a set of dyadic boxes $B$ such that for a given dyadic box $b$, queries that return the set $\{b'\in B: b\subseteq b'\}$ can be computed in $\tilde{O}(1)$ time. In this data structure, insertions and deletions of single dyadic boxes can also be done in $\tilde{O}(1)$ time. We will assume our MDBCIs are stored in data structures with this property.

Consider the runtime of InsertMDBCI. Lines~\ref{mdbci-insert:inner} and~\ref{mdbci-insert:insert} modify one prefix of a box and insert the new box into $B'$, both of which can be done in $\tilde{O}(1)$ time. The inner loop on line~\ref{mdbci-insert:third-loop} iterates over at most $d=\tilde{O}(1)$ attributes. The second loop on line~\ref{mdbci-insert:second-loop} iterates over at most $\tilde{O}(1)$ boxes, by Lemma~\ref{lemma:dyadic}. Because $B$ is stored in the dyadic box data structure mentioned above, the deletion of a box on line~\ref{mdbci-insert:remove} takes $\tilde{O}(1)$ time, and the outer loop on line~\ref{mdbci-insert:first-loop} can find the set of at most $\tilde{O}(1)$ boxes to iterate over in $\tilde{O}(1)$ time. In total, InsertMDBCI runs in $\tilde{O}(1)$ time.

Now consider the runtime of DeleteMDBCI. Lines~\ref{mdbci-delete:inner} and~\ref{mdbci-delete:box-contains} modify one prefix of a box and then query $B$ for any boxes containing $b''$, which can be answered in $\tilde{O}(1)$ time by the aforementioned dyadic box data structure. The inner loop on line~\ref{mdbci-delete:third-loop} iterates over at most $d=\tilde{O}(1)$ attributes. The second loop on line~\ref{mdbci-delete:second-loop} iterates over at most $\tilde{O}(1)$ boxes by Lemma~\ref{lemma:dyadic}. The insertion of a single box on line~\ref{mdbci-delete:add} takes $\tilde{O}(1)$ time. The outer loop on line~\ref{mdbci-delete:first-loop} iterates over $\tilde{O}(1)$ boxes by Lemma~\ref{lemma:dyadic}. In total, DeleteMDBCI runs in $\tilde{O}(1)$ time.
\end{proof}
}{} 
} 


\input{nph-proof}

\vspace{-10pt}
\section{ADORA's Runtime Analysis}
\label{app:adora-runtime}
ADORA calls Algorithm~\ref{alg:order-attr} $n$, so $\tilde{O}(1)$, times. In  Algorithm~\ref{alg:order-attr}, the sorting of $m$ relations according to $\phi$ on line~\ref{line:subroutine:firstfor} takes $\tilde{O}(N)$ time. 
The for-loop beginning on line~\ref{line:subroutinesecondfor} iterates over each domain value $a\in D$ and each $R\in\mathcal{S}$ and
appends $H(R,A,a)$ to $\mathcal{T}[a]$. Since $R$ was sorted lexicographically according to $\phi$, which places $A$ as the first relation, all tuples with the same $A$-value are now consecutive in $R$. Therefore, with a single
linear pass through $R$, we can compute all of the hyperplanes $H(R,A,a)$. We do this for each relation,
so the runtime is bounded by $O(mN)=\tilde{O}(N)$. For the final sorting of $D$ on line~\ref{line:subroutinefinalsort} observe that the total size of the array $\mathcal{T}$, summed over all domain values $a$, is at most $N$. Thus, we are sorting an array of arrays where the total amount
of data is of size $\tilde{O}(N)$, which can be done  in $\tilde{O}(N)$ time (e.g., with a merge-sort algorithm that merges two sorted sub-arrays in $\tilde{O}(N)$ time), completing the proof that ADORA's runtime is $\tilde{O}(N)$ as claimed in Theorem~\ref{thm:domain-ordering-approx}.

\vspace{-10pt}
\section{The ADORA Approximation Bound Is Tight}
\label{sec:adora-tight}

Theorem~\ref{thm:domain-ordering-approx} proved that ADORA produces a domain ordering $\sigma$ for $Q$ such that
$K_{\Box}(\sigma(Q))=\tilde{O}(K^r)$, where $K$ is the minimum box cover size under any domain ordering
and $r$ is the maximum arity of a relation in $Q$. We will show that this bound is
tight by presenting a class of 2D relations $R_d$ for which ADORA returns a domain ordering $\sigma$
such that $K_{\Box}(\sigma(R_d))=\Omega(K^2)$, where $K$ is the minimum box cover size for $R_d$ under any ordering.
For any integer $d>0$, let $R_d(A,B)$ be the relation over 2 $d$-bit attributes $A$ and $B$ given by
\begin{center}
$R_d(A,B)=\{\langle 0a,0b\rangle:a,b\in\{0,1\}^{d-1}, a\neq b\}\cup\{\langle 1a,1b\rangle:a,b\in\{0,1\}^{d-1},a\neq b\}$
\end{center}

\iftoggle{appendix-online}{
\input{figs/ex-adora-tight}

The relation $R_3$ is depicted in Figure~\ref{fig:adora-tight} (left).
The minimum size box cover for $R_3$ consists of the 2 boxes which cover the top left and bottom right quadrants,
the $2\times 2$ box which covers the middle 4 cells,
as well as the 6 unit boxes which cover the diagonal line from the bottom left to the top right, for
a total box cover size of 9. This happens to be the minimum box cover size for $R_3$ under any domain ordering.
The relation $\sigma(R_3)$ depicted in Figure~\ref{fig:adora-tight} (right) is $R_3$
under a different domain ordering $\sigma$, obtained by moving all of the even
domain values to the range $[000-011]$ and all of the odd domain values to the range $[100-111]$ in $A$ and $B$.
The minimum box cover for $\sigma(R_3)$ consists of the 18 unit boxes covering the gap cells which are surrounded
by tuples, plus the 7 $2\times 2$ boxes which can be tiled to cover the diagonal stretch of gaps, for a total
box cover size of 25.
$R_3$ generalizes to $R_d$ for any $d>0$.
}{
\kaleb{As an example, $R_3$ is illustrated in Appendix~E of the online version of this paper~\cite{this-arxiv}.}
}
The default ordering of $R_d$ has a minimum box cover of size $K=2^d+1$. However, there is
a bad ordering $\sigma_d$ such that $\sigma_d(R_d)$ has a minimum box cover size of
\iftoggle{appendix-online}{$2^d\cdot 2^{d-1}-2^d+1=$}{}$\Omega(2^{2d})=\Omega(K^2)$. The key observation about this example is that no rows or columns in $R_d$ are equal, so ADORA may return $\sigma_d$ as a solution.
Since $R_d$ has arity 2, the bound of Theorem~\ref{thm:domain-ordering-approx}
is tight in this case.

\vspace{-10pt}
\section{Approximating $\DomOrCert$ On Fully Semi-join Reduced Queries}
\label{sec:semi-join-certificate}

This section serves to illustrate that if the input relations of a query $Q$ are fully semi-join reduced, so we know a priori that all of the input tuples contribute to the query's output, then  $\DomOrCert$ can be approximated with ADORA. 
We use the term \em{``dangling'' (input) tuple} as follows. Assume the domain ordering $\sigma$ is fixed. Given a query $Q=(\mathcal{R},\mathcal{A})$ (under $\sigma$) and a relation $R\in\mathcal{R}$, the tuple $t\in R$ is a
\em{dangling tuple} if there is no tuple $t'$ in the output of $Q$ such that $\pi_{\attr(R)}(t')=t$.
$Q$ is said to be fully semi-join reduced~\cite{bernstein:semi-joins} if there are no dangling tuples in any of the relations in $\mathcal{R}$.
The problem of fully semi-join reducing a query by removing all of the dangling tuples
is known to be hard for cyclic queries~\cite{bernstein:semi-joins}.
\iftoggle{appendix-online}{We next show that an oracle which computes the full semi-join reduction of a query $Q$ would allow us to
bridge the gap between minimizing the box cover size and certificate size for $Q$.}
{
\kaleb{Proposition~\ref{propn:semi-join} (proven in Appendix~F of the online version of this paper~\cite{this-arxiv}) shows that an oracle which computes the full semi-join reduction of a query $Q$ would allow us to
bridge the gap between minimizing the box cover size and certificate size for $Q$.}}

\begin{proposition}
\label{propn:semi-join}
Let $\sigma$ be a domain ordering and let \iftoggle{appendix-online}{$Q=(\mathcal{R},\mathcal{A})$}{$Q$} be
 a fully semi-join reduced query under $\sigma$. Let  $K_{\Box}(Q)$ be the size of the minimum box cover for $Q$ under $\sigma$. Let $C_{\Box}(Q)$ be the size of the minimum certificate of $Q$ under $\sigma$. Then $K_{\Box}(Q)=\tilde{\Theta}(C_{\Box}(Q))$.
\end{proposition}

\iftoggle{appendix-online}{
\begin{proof}
Let $\mathcal{C}$ be a box certificate for $Q$ of size $C_{\Box}(Q)$.
Let $b\in\mathcal{C}$ and $R\in\mathcal{R}$.
Since $b$ is a box in the certificate, all of the tuples contained in $b$
must not be part of the output of $Q$. Since $Q$ has no dangling tuples,
the projection $b'$ of $b$ onto the attributes of $R$ must not contain any tuples
of $R$, otherwise these would be dangling tuples.
Thus $b'$ is a gap box for $R$.

Let $\mathcal{B}_R$ be the set of all such projections
$b'$ of boxes in $\mathcal{C}$ onto the attributes of $R$.
We claim that $\mathcal{B}_R$ forms a box cover for $R$.
Indeed, for any tuple $t'\not\in R$, there exists some tuple
$t$ not in the output of $Q$ such that $\pi_{\attr(R)}(t)=t'$.
Since $t$ is not in the output of $Q$, there exists a box $b\in\mathcal{C}$
which contains $t$, and therefore the corresponding projection box
$b'\in\mathcal{B}_R$ covers $t'$.

This demonstrates that the minimum box cover size for the
relation $R$ is at most $|\mathcal{C}|=C_{\Box}(Q)$.
Since $R$ was an arbitrary relation in $\mathcal{R}$, we can repeat
this process for all other relations in $\mathcal{R}$ to obtain a box
cover for the entire query $Q$ of size $m\cdot |\mathcal{C}|=\tilde{O}(C_{\Box}(Q))$,
so $K_{\Box}(Q)=\tilde{O}(C_{\Box}(Q))$.
By definition, we also have $C_{\Box}(Q)\leq K_{\Box}(Q)$, so this proves the proposition.
\end{proof}

Therefore, Proposition~\ref{propn:semi-join} and Theorem~\ref{thm:domain-ordering-approx} imply that in the special case when $Q$ is fully semi-join reduced, we can use ADORA to obtain a domain ordering $\sigma_{ADORA}$, under which the certificate size would be at most $\tilde{O}(C_{\Box}(Q)^r)$. The proof above relies on the fact that when there are no dangling tuples in the query,
a box certificate for the query immediately yields a box cover for each of the base relations
of the same size. In general queries, the dangling tuples in each of the base relations may form
arbitrarily complex shapes which can make the minimum box cover size much larger than the minimum
box certificate size.
} 

\iftoggle{appendix-online}{
\vspace{-10pt}
\section{Generating a Certificate with Tetris}
\label{sec:tetris-certificate}

It is possible to modify Tetris so that it computes an approximately minimum size
box certificate for $Q$ as it computes the output for $Q$. Given input box cover $\mathcal{B}$,
this simple modification to Tetris will compute a box certificate for $\mathcal{B}$
of size $\tilde{O}(C_{\Box}(\mathcal{B}))$.

We reviewed Tetris briefly in Section~\ref{sec:preliminaries}.
In particular, in this section we will focus on the TetrisReloaded variant,
which initializes its knowledge base of boxes to be empty, then adds boxes to the knowledge base
whenever its subroutine TetrisSkeleton performs a geometric resolution or returns a witness tuple $o$ not covered by a box in the knowledge base.
We defer to reference~\cite{tetris} for a detailed description of TetrisReloaded.

Let $\mathcal{B}$ be the original box cover input to Tetris, and let $\mathcal{K}$
be the knowledge base of gap boxes that Tetris initializes as empty.
As our modification to Tetris, we will add a new set of boxes $\mathcal{C}$ which we initialize as empty.
TetrisSkeleton returns YES if the current knowledge base covers the entire output space,
or it returns a witness tuple $o$ otherwise.
Tetris then checks if $o$ is an output tuple by querying $\mathcal{B}$ for any gap boxes
which contain $o$.
If $\mathcal{B}_o\subseteq\mathcal{B}$ is the set of boxes in $\mathcal{B}$ which contain $o$,
and $\mathcal{B}_o\neq\emptyset$, then $o$ is a gap tuple, so Tetris sets $\mathcal{K}:=\mathcal{K}\cup\mathcal{B}_o$.
At this point, we modify Tetris once again by also setting $\mathcal{C}:=\mathcal{C}\cup\mathcal{B}_o$.
If $\mathcal{B}_o = \emptyset$, then $o$ is an output tuple, so Tetris outputs $o$ and inserts $o$ as a unit gap box into $\mathcal{K}$.
This process repeats until the boxes in $\mathcal{K}$ cover the entire output space.

After our modified Tetris finishes executing, the resulting set $\mathcal{C}$ must form a certificate for $\mathcal{B}$,
because if there is any gap tuple not covered by $\mathcal{C}$, Tetris would have encountered it as a witness
before finishing. Let $W$ be the set of witness gap tuples Tetris encountered which resulted in adding one or
more boxes to $\mathcal{C}$. Then every pair of witnesses $o_1,o_2\in W$ must be independent in the sense
that there is no box $b$ in $\mathcal{B}$ that covers both $o_1$ and $o_2$. Otherwise, if $o_1$ was encountered first,
then $b$ would have been in $\mathcal{K}$ already when $o_2$ was returned by TetrisSkeleton, which is a contradiction.
This implies that any certificate for $\mathcal{B}$ must have size at least $|W|$. By Lemma~\ref{lemma:dyadic},
we also have that $\mathcal{C}$ has size at most $\tilde{O}(|W|)$, since $|\mathcal{B}_o|=\tilde{O}(1)$ for
each $o\in W$. Therefore $|\mathcal{C}|=\tilde{O}(C_{\Box}(\mathcal{B})|$, i.e. $\mathcal{C}$ is a $\tilde{O}(1)$
factor approximation of the minimum certificate for $\mathcal{B}$.

}{}

%% file: figs/ex-many-general-boxes.tex
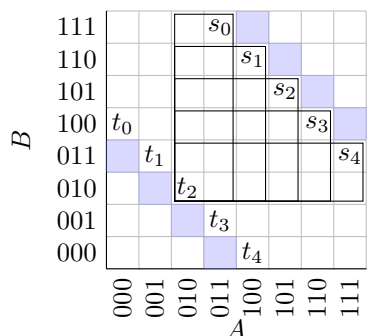
\begin{figure}[t]
\vspace{-10pt}
\begin{center}
\begin{tikzpicture}
\begin{axis}[axis lines=left,xlabel=$A$,ylabel=$B$,
  xmin=0,xmax=8,
  ymin=0,ymax=8,
  tickwidth=0,
  xticklabel interval boundaries,yticklabel interval boundaries,
  xtick={0,1,2,3,4,5,6,7,8},ytick={0,1,2,3,4,5,6,7,8},
  xticklabels={000,001,010,011,100,101,110,111},
  yticklabels={000,001,010,011,100,101,110,111},
  xticklabel style={rotate=90},
  y label style={at={(axis description cs:-0.04,0.5)},anchor=north},
  x label style={at={(axis description cs:0.5,-0.15)},anchor=south},
  grid=both,axis line style={-},
  height=5cm,width=5cm]
\filldraw[\tuplecolour,opacity=\tupleopacity] (axis cs:0,3) rectangle (axis cs:1,4);
\filldraw[\tuplecolour,opacity=\tupleopacity] (axis cs:1,2) rectangle (axis cs:2,3);
\filldraw[\tuplecolour,opacity=\tupleopacity] (axis cs:2,1) rectangle (axis cs:3,2);
\filldraw[\tuplecolour,opacity=\tupleopacity] (axis cs:3,0) rectangle (axis cs:4,1);
\filldraw[\tuplecolour,opacity=\tupleopacity] (axis cs:4,7) rectangle (axis cs:5,8);
\filldraw[\tuplecolour,opacity=\tupleopacity] (axis cs:5,6) rectangle (axis cs:6,7);
\filldraw[\tuplecolour,opacity=\tupleopacity] (axis cs:6,5) rectangle (axis cs:7,6);
\filldraw[\tuplecolour,opacity=\tupleopacity] (axis cs:7,4) rectangle (axis cs:8,5);
\node at (axis cs:0.5,4.5) {$t_0$};
\node at (axis cs:1.5,3.5) {$t_1$};
\node at (axis cs:2.5,2.5) {$t_2$};
\node at (axis cs:3.5,1.5) {$t_3$};
\node at (axis cs:4.5,0.5) {$t_4$};
\node at (axis cs:3.5,7.5) {$s_0$};
\node at (axis cs:4.5,6.5) {$s_1$};
\node at (axis cs:5.5,5.5) {$s_2$};
\node at (axis cs:6.5,4.5) {$s_3$};
\node at (axis cs:7.5,3.5) {$s_4$};
\draw[\boxcolour] (axis cs:2.1,2.1) rectangle (axis cs:3.9,7.9);
\draw[\boxcolour] (axis cs:2.1,2.1) rectangle (axis cs:4.9,6.9);
\draw[\boxcolour] (axis cs:2.1,2.1) rectangle (axis cs:5.9,5.9);
\draw[\boxcolour] (axis cs:2.1,2.1) rectangle (axis cs:6.9,4.9);
\draw[\boxcolour] (axis cs:2.1,2.1) rectangle (axis cs:7.9,3.9);
\end{axis}
\end{tikzpicture}
\end{center}
\vspace{-20pt}
\caption{Example of a relation $R_8(A,B)$ with $\omega(N)$ maximal general gap boxes.}
\label{fig:many-general-boxes}
\vspace{-10pt}
\end{figure}

%% file: nph-proof.tex
\vspace{-10pt}
\section{Proof that \texorpdfstring{$\DomOrBox$}{DomOrBoxMinB} is NP-hard}
\label{sec:nph-full}
This section contains proofs of the 6 transformation steps we used in the proof of Theorem~\ref{thm:nph}.

\vspace{-10pt}
\subsection{Proof of Step 1:}
\label{app:step1}

\textbf{Claim.} \emph{Every $r_{i,j}$ row can be made adjacent to some equivalent $r_{k,\ell}$ row.}

Let $r_1:=r_{i,j}$ be a row which is not adjacent to any equivalent row.
Let $r_2:=r_{k,\ell}$ be any row equivalent to $r_1$ (at least one such row exists because we duplicate each row of $M$ when constructing $M'$). Since $r_1$ is not adjacent to
any equivalent row and there are an even number of rows equivalent to $r_1$, there must
be some run $E$ of rows equivalent to $r_1$ with odd length. If $E$ has length 1,
we assume $r_2$ is the one row in $E$, and therefore $r_2$ is not adjacent to any equivalent row.
If $E$ has length at least 3, we assume $r_2$
is the second row in $E$, and therefore $r_2$ is not adjacent to $p_{k,\ell}$.
Let $p_1:=p_{i,j}$ and let $p_2:=p_{k,\ell}$.
Let $c_{p1}$ be the column where $p_1$ has a 1-cell, and let $c_{p2}$ be the column where $p_2$ has a 1-cell.
Let $c_1$ and $c_2$ be the columns where $r_1$ and $r_2$ both have 1-cells.
Let $b_1\in B$ be the box covering the padding column in $r_1$ with greatest width.
Let $b_2\in B$ be the box covering the padding column in $r_2$ with greatest width.
Let $b_3\in B$ be the box covering the padding column in $p_1$.
Let $b_4\in B$ be the box covering the padding column in $p_2$.
Our approach will be to remove the rows $r_1,r_2,p_1$, and $p_2$ from
$M'$, then insert them in the order $(p_1,r_1,r_2,p_2)$ at the bottom of $M'$.
In this order, the 1-cells of these rows can be covered by 4 boxes,
regardless of the column ordering. A box of width 1 and height 2 can be used to cover the two 1-cells
in each of the columns in $\{c_1,c_2,c_{p1},c_{p2}\}$.
To show that this modification does not increase
the number of boxes in $B$, it suffices to show that there are at least 4 boxes
which can be removed from $B$ when we remove these 4 rows from $M'$.
We split our analysis into four cases.

\begin{enumerate}
\item $b_1\neq b_3$ and $b_2\neq b_4$.
In this case, all of $\{b_1,b_2,b_3,b_4\}$ are distinct and all 4 of these boxes are removed
when we remove the rows $r_1,r_2,p_1,p_2$.
\item $b_1\neq b_3$ and $b_2=b_4$.
Since $b_2=b_4$, $r_2$ is adjacent to $p_2$. By our previous assumptions about $r_2$, $r_2$ is not adjacent to any equivalent row. W.l.o.g., assume that $p_2$ is directly
below $r_2$. Let $r_3$ be the row directly above $r_2$. $r_3$ is not equivalent to $r_2$,
so there exists a box $b_5$ covering at least one of $c_1$ or $c_2$ in $r_2$ with height 1,
since it cannot extend vertically to either $p_2$ or $r_3$.
$b_5$ is not equal to $b_2$, because $b_2$ has height 2.
$\{b_1,b_2,b_3,b_5\}$ is a set of 4 distinct boxes which are removed when we remove
the rows $\{r_1,r_2,p_1,p_2\}$.
\item $b_1=b_3$ and $b_2\neq b_4$.
Since $b_1=b_3$, $r_1$ is adjacent to $p_1$. Suppose w.l.o.g. that
$p_1$ is directly above $r_1$. Let $r_3$ be the row directly below $r_1$. Since $r_1$
is not adjacent to any equivalent rows, $r_3$ is not equivalent to $r_1$.
Therefore, there is a box $b_6\in B$ covering at least one of $c_1$ or $c_2$ in $r_1$
which has height 1, since it cannot extend vertically to either $p_1$ or $r_3$.
$b_6$ is not equal to $b_1$, since $b_1$ has height 2.
Now the set of boxes $\{b_1,b_2,b_4,b_6\}$ is a set of 4 distinct boxes which are removed
when we remove the rows $\{r_1,r_2,p_1,p_2\}$.
\item $b_1=b_3$ and $b_2=b_4$.
This case can be proven by combining the arguments from the previous two cases.
Since $b_2=b_4$, we can define the box $b_5$ exactly as in case 2.
Since $b_1=b_3$, we can define the box $b_6$ exactly as in case 3.
Then, $\{b_1,b_2,b_5,b_6\}$ is a set of 4 distinct boxes which are removed from $B$
when we remove the rows $\{r_1,r_2,p_1,p_2\}$.
\end{enumerate}

\vspace{-10pt}
\subsection{Proof of Step 2}
\label{app:step2}

\textbf{Claim.} \emph{Every run of equivalent $r_{i,j}$ rows can be made to have even length.}

Let $E_1$ be a run of equivalent $r_{i,j}$ rows of odd length.
By the claim of step 1, $E_1$ has length $\geq 3$.
Let $r_1$ be the second row in $E_1$.
Since $E_1$ has odd length and there are an even number of rows equivalent to $r_1$,
there exists another run $E_2$ of rows equivalent to $r_1$ with odd length.
$E_2$ also has length $\geq 3$.
Let $c_{p1}$ be the column which has a 1-cell only in $r_1$ and its corresponding padding row.
Let $b\in B$ be the box which covers $c_{p1}$ in $r_1$. Since $r_1$ is not adjacent to its
padding row, $b$ has height 1.
If we remove $r_1$ from $M'$, $b$ can be removed. 
By inserting $r_1$ directly below the first row in $E_2$, a unit box can be used to cover $c_{p1}$ in
$r_1$.
Let $r_2$ be the first row in $E_2$. Let $c_1$ and $c_2$ be the two columns of $M'$ where $r_1$ and $r_2$
share 1-cells.
To cover these two 1-cells in
$r_1$, we can extend vertically the boxes covering $c_1$ and $c_2$ in $r_2$. We may assume these
boxes can be extended vertically, because at most two of the rows in $E_2$ have their
$c_1$ (or $c_2$) cell covered by a box which stretches horizontally from a padding column.
That is, there is \emph{some} row in $E_2$ where the box covering the $c_1$ (or $c_2$) cell can
be extended vertically to cover the $c_1$ (or $c_2$) cell of $r_1$.
Hence, this transformation can be made without increasing the size of $B$.
After this, both $E_1$ and $E_2$ have even length. Continue this process until every run of equivalent $r_{i,j}$ rows have even length.

\vspace{-10pt}
\subsection{Proof of Step 3}
\label{app:step3}

\textbf{Claim.} \emph{Every run of equivalent $r_{i,j}$ rows can be made to have length 2.}

Let $E$ be a run of equivalent $r_{i,j}$ rows of even length greater than 2. So $E$ has a length of 
at least 4. Let $r_1$ be the second row in $E$
and let $r_2$ be the third row in $E$. Since $E$ has length at least 4, neither $r_1$ nor
$r_2$ are adjacent to their respective padding rows, $p_1$ and $p_2$.
We claim the boxes covering the padding columns in $r_1$ and $r_2$
have width 1. We split our analysis into two cases. Let $c_1$ and $c_2$ be the two columns where
$r_1$ and $r_2$ both have 1-cells.
\begin{enumerate}
\item $c_1$ and $c_2$ are adjacent.
At most 2 of the rows in $E$ have their padding columns adjacent to $(c_1,c_2)$ on either side.
This means there is some row $r_3$ in $E$ where the box $b$ covering $c_1$ and $c_2$ does not
also cover its padding column. $b$ can be extended vertically to cover $c_1$ and $c_2$ in all rows
of $E$. Any boxes covering padding columns for rows in $E$ can be replaced with boxes of width 1,
and all of the 1-cells in the rows of $E$ remain covered.
\item $c_1$ and $c_2$ are not adjacent.
At most 2 rows in $E$ have their padding columns adjacent to $c_1$ on either side.
This means there is some row $r_3$ in $E$ where the box $b$ covering $c_1$ does not
also cover its padding column. $b$ can be extended vertically to cover $c_1$ in all rows of $E$.
The same argument applies for $c_2$. Any boxes covering padding columns for rows in $E$ can be
replaced with boxes of width 1, and all 1-cells in the rows of $E$ remain covered.
\end{enumerate}

Now, removing $p_1$ and $p_2$ removes two boxes from $B$, since unit boxes must be covering the single 1-cells
in $p_1$ and $p_2$.
Inserting $(p_1,p_2)$ in order in between $r_1$ and $r_2$, we can cover the 1-cells in $(c_{p1},p_1)$
and $(c_{p2},p_2)$ by extending vertically the width 1 boxes covering $(c_{p1},r_1)$ and $(c_{p2},r_2)$.
This splits any boxes which vertically streched from $r_1$ to $r_2$ into two. There were at most two
such boxes, so the total number of boxes in $B$ does not increase. Now $E$ is split into two distinct runs
of equivalent rows, one of length 2 and one of length $|E|-2$. This process can be
repeated until all runs have length exactly 2.

\vspace{-10pt}
\subsection{Proof of Step 4}
\label{app:step4}

\textbf{Claim.} \emph{The padding rows $p_{i,j}$ can be made adjacent to their matching $r_{i,j}$ rows.}

Let $r_1:=r_{i,j}$ be a row which is not adjacent to its padding row $p_1:=p_{i,j}$.
By the claim of step 3, we know $r_1$ is adjacent to exactly one row, $r_2$, that is equivalent to $r_1$.
Let $p_2$ be the padding row matching $r_2$.
Let $c_1$ and $c_2$ be the columns where $r_1$ and $r_2$ share 1-cells.
Let $c_{p1}$ be the column which has 1-cells only in $r_1$ and $p_1$.
Let $c_{p2}$ be the column which has 1-cells only in $r_2$ and $p_2$.
Let $b_1$ be the box which covers the 1-cell in row $r_1$ and column $c_{p1}$
of greatest width.
Let $b_2$ be the box which covers the 1-cell in row $r_2$ and column $c_{p2}$
of greatest width.
Let $b_3$ be the box which covers the 1-cell in $p_1$.
Let $b_4$ be the box which covers the 1-cell in $p_2$.
We split our analysis into two cases.
\begin{enumerate}
\item $r_2$ is adjacent to $p_2$.
In this case, similar to our argument in step 1, there exists a box $b_5\in B$
with height 1 which covers $c_1$ or $c_2$ (or both) in $r_2$.
By removing the rows $\{r_1,r_2,p_1,p_2\}$, the 4 distinct boxes $\{b_1,b_2,b_3,b_5\}$
are all removed from $B$.
By inserting the rows $(p_1,r_1,r_2,p_2)$ in order at the bottom of the matrix,
we can cover their 1-cells with at most 4 boxes, so the total number of boxes in $B$ does not increase.
\item $r_2$ is not adjacent to $p_2$.
In this case, $r_1$ is not adjacent to $p_1$ and $r_2$ is not adjacent to $p_2$, so
$\{b_1,b_2,b_3,b_4\}$ are 4 distinct boxes in $B$ which are removed if we remove
rows $\{r_1,r_2,p_1,p_2\}$. By inserting $(p_1,r_1,r_2,p_2)$ in order at the bottom
of the matrix, we can cover their 1-cells with at most 4 boxes, so the size of $B$ does not increase.
\end{enumerate}

We can repeat this process until all $r_{i,j}$ rows are adjacent to their matching $p_{i,j}$ rows.

\vspace{-10pt}
\subsection{Proof of Step 5}
\label{app:step5}

\textbf{Claim.} \emph{The row order $\sigma_r'$ can be made to exactly match the default row order of $M'$.}

By the claims of steps 3 and 4, all of the rows are now divided into separate 4-row units
containing a run of two equivalent $r_{i,j}$ rows surrounded by their two matching padding
rows. There are no boxes in $B$ which can stretch vertically across two or more of these separate
units, because there are no two $p_{i,j}$ rows which share a 1-cell. Thus, we are free to reorder
these units arbitrarily. Order the units so that for all $i$, the $i$-th unit contains two $r_{i,j}$ rows which
correspond to the $i$-th row of the original matrix $M$. The resulting row order $\sigma_r'$ is then
equal to the default row ordering of $M'$, modulo any equivalent rows which are swapped from their default positions.
Since equivalent rows are equal up to reordering the columns of $M'$, there exists an ordering on
the columns of $M'$ that transforms $\sigma_r'(M')$ back to the original matrix $M'$.
In other words, the row ordering $\sigma_r'$ is now equivalent to the default row ordering of $M'$
up to a relabelling of the rows. This is sufficient for our purposes, since we can relabel the rows
accordingly and move on to modifying the column ordering only.

\vspace{-10pt}
\subsection{Proof of Step 6}
\label{app:step6}

\textbf{Claim.} \emph{The column order $\sigma_c'$ can be made to exactly match the default column order of $M'$
on the last $2n$ columns.}

For each padding row $p_{i,j}$, the box $b$ covering the single 1-cell in $p_{i,j}$ has width 1.
By step 4, each padding row is adjacent to its corresponding $r_{i,j}$ row. This means $b$ extends
vertically to also cover the only other 1-cell in its column. Therefore, by moving this column to the right side of the matrix,
we do not increase the total number of boxes in $B$.
Once all of these padding columns have been moved to the right, the boxes covering all of their 1-cells
all have width 1. Thus, we can reorder them to exactly match the last $2n$ columns in the default column ordering
of $M'$ without modifying any boxes in $B$.

%% file: figs/ex-adora-tight.tex
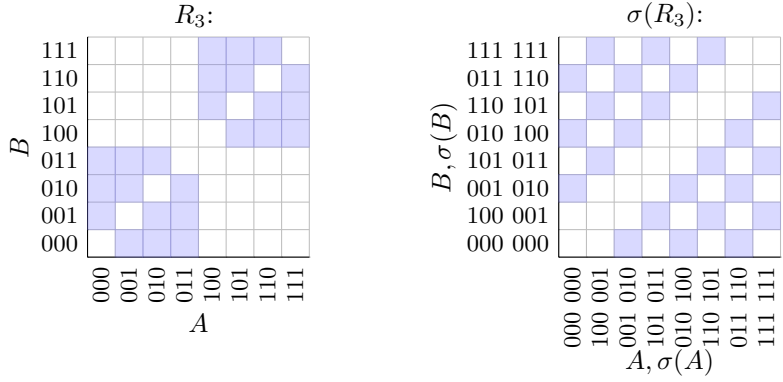
\begin{figure*}[t]

\begin{center}
\vspace{-15pt}
\begin{tikzpicture}
\node at (1.2cm,3.2cm) {$R_3$:};
\node at (7.4cm,3.2cm) {$\sigma(R_3)$:};
\begin{axis}[axis lines=left,xlabel=$A$,ylabel=$B$,
  xmin=0,xmax=8,
  ymin=0,ymax=8,
  tickwidth=0,
  xticklabel interval boundaries,yticklabel interval boundaries,
  xtick={0,1,2,3,4,5,6,7,8},ytick={0,1,2,3,4,5,6,7,8},
  xticklabels={000,001,010,011,100,101,110,111},yticklabels={000,001,010,011,100,101,110,111},
  xticklabel style={rotate=90,font=\small},
  yticklabel style={font=\small},
  y label style={at={(axis description cs:0.03,0.5)},anchor=north},
  x label style={at={(axis description cs:0.5,-0.2)},anchor=south},
  grid=both,axis line style={-},
  height=4.5cm,width=4.5cm,
  at={(-0.2cm,0cm)}]
\filldraw[\tuplecolour,opacity=\tupleopacity] (axis cs:1,0) rectangle (axis cs:2,1);
\filldraw[\tuplecolour,opacity=\tupleopacity] (axis cs:2,0) rectangle (axis cs:3,2);
\filldraw[\tuplecolour,opacity=\tupleopacity] (axis cs:3,0) rectangle (axis cs:4,3);
\filldraw[\tuplecolour,opacity=\tupleopacity] (axis cs:0,1) rectangle (axis cs:1,4);
\filldraw[\tuplecolour,opacity=\tupleopacity] (axis cs:1,2) rectangle (axis cs:2,4);
\filldraw[\tuplecolour,opacity=\tupleopacity] (axis cs:2,3) rectangle (axis cs:3,4);
\filldraw[\tuplecolour,opacity=\tupleopacity] (axis cs:5,4) rectangle (axis cs:6,5);
\filldraw[\tuplecolour,opacity=\tupleopacity] (axis cs:6,4) rectangle (axis cs:7,6);
\filldraw[\tuplecolour,opacity=\tupleopacity] (axis cs:7,4) rectangle (axis cs:8,7);
\filldraw[\tuplecolour,opacity=\tupleopacity] (axis cs:4,5) rectangle (axis cs:5,8);
\filldraw[\tuplecolour,opacity=\tupleopacity] (axis cs:5,6) rectangle (axis cs:6,8);
\filldraw[\tuplecolour,opacity=\tupleopacity] (axis cs:6,7) rectangle (axis cs:7,8);
\end{axis}
\begin{axis}[axis lines=left,xlabel={$A,\sigma(A)$},ylabel={$B,\sigma(B)$},
  xmin=0,xmax=8,
  ymin=0,ymax=8,
  tickwidth=0,
  xticklabel interval boundaries,yticklabel interval boundaries,
  xtick={0,1,2,3,4,5,6,7,8},ytick={0,1,2,3,4,5,6,7,8},
  xticklabels={000 000,100 001,001 010,101 011,010 100,110 101,011 110,111 111},
  yticklabels={000 000,100 001,001 010,101 011,010 100,110 101,011 110,111 111},
  xticklabel style={rotate=90,font=\small},
  yticklabel style={font=\small},
  y label style={at={(axis description cs:-0.2,0.5)},anchor=north},
  x label style={at={(axis description cs:0.5,-0.4)},anchor=south},
  grid=both,axis line style={-},
  height=4.5cm,width=4.5cm,
  at={(6cm,0cm)}]
\filldraw[\tuplecolour,opacity=\tupleopacity] (axis cs:2,0) rectangle (axis cs:3,1);
\filldraw[\tuplecolour,opacity=\tupleopacity] (axis cs:4,0) rectangle (axis cs:5,1);
\filldraw[\tuplecolour,opacity=\tupleopacity] (axis cs:6,0) rectangle (axis cs:7,1);
\filldraw[\tuplecolour,opacity=\tupleopacity] (axis cs:3,1) rectangle (axis cs:4,2);
\filldraw[\tuplecolour,opacity=\tupleopacity] (axis cs:5,1) rectangle (axis cs:6,2);
\filldraw[\tuplecolour,opacity=\tupleopacity] (axis cs:7,1) rectangle (axis cs:8,2);
\filldraw[\tuplecolour,opacity=\tupleopacity] (axis cs:0,2) rectangle (axis cs:1,3);
\filldraw[\tuplecolour,opacity=\tupleopacity] (axis cs:4,2) rectangle (axis cs:5,3);
\filldraw[\tuplecolour,opacity=\tupleopacity] (axis cs:6,2) rectangle (axis cs:7,3);
\filldraw[\tuplecolour,opacity=\tupleopacity] (axis cs:1,3) rectangle (axis cs:2,4);
\filldraw[\tuplecolour,opacity=\tupleopacity] (axis cs:5,3) rectangle (axis cs:6,4);
\filldraw[\tuplecolour,opacity=\tupleopacity] (axis cs:7,3) rectangle (axis cs:8,4);
\filldraw[\tuplecolour,opacity=\tupleopacity] (axis cs:0,4) rectangle (axis cs:1,5);
\filldraw[\tuplecolour,opacity=\tupleopacity] (axis cs:2,4) rectangle (axis cs:3,5);
\filldraw[\tuplecolour,opacity=\tupleopacity] (axis cs:6,4) rectangle (axis cs:7,5);
\filldraw[\tuplecolour,opacity=\tupleopacity] (axis cs:1,5) rectangle (axis cs:2,6);
\filldraw[\tuplecolour,opacity=\tupleopacity] (axis cs:3,5) rectangle (axis cs:4,6);
\filldraw[\tuplecolour,opacity=\tupleopacity] (axis cs:7,5) rectangle (axis cs:8,6);
\filldraw[\tuplecolour,opacity=\tupleopacity] (axis cs:0,6) rectangle (axis cs:1,7);
\filldraw[\tuplecolour,opacity=\tupleopacity] (axis cs:2,6) rectangle (axis cs:3,7);
\filldraw[\tuplecolour,opacity=\tupleopacity] (axis cs:4,6) rectangle (axis cs:5,7);
\filldraw[\tuplecolour,opacity=\tupleopacity] (axis cs:1,7) rectangle (axis cs:2,8);
\filldraw[\tuplecolour,opacity=\tupleopacity] (axis cs:3,7) rectangle (axis cs:4,8);
\filldraw[\tuplecolour,opacity=\tupleopacity] (axis cs:5,7) rectangle (axis cs:6,8);
\end{axis}
\end{tikzpicture}
\end{center}
\vspace{-20pt}
\caption{A relation $R_3$ for which the bound of Theorem~\ref{thm:domain-ordering-approx} is tight.}
\label{fig:adora-tight}
\vspace{-10pt}
\end{figure*}